\documentclass[nofootinbib,longbibliography]{article}
\usepackage{graphicx} 
\usepackage{graphicx}
\usepackage{amsmath}
\usepackage{amssymb}
\usepackage{amsfonts}
\usepackage{amsthm}
\usepackage{subfigure}
\usepackage{dsfont}
\usepackage{appendix}
\usepackage{tikz}
\usepackage{braket}
\usepackage{listings}
\usepackage[margin=0.8in]{geometry}
\usepackage{comment}
\usepackage{array}
\usepackage[normalem]{ulem}
\usepackage{float}
\usepackage{appendix}
\usepackage{slashed}
\usepackage[affil-it]{authblk}

\usepackage[
  backend=biber,
  style=phys,
  articletitle=true,  
  biblabel=brackets,  
  chaptertitle=false,
  pageranges=true
]{biblatex}
\newtheorem{theorem}{Theorem}
\newtheorem*{theorem*}{Theorem}

\newtheorem{lemma}[theorem]{Lemma}
\newtheorem*{lemma*}{Lemma}
\newtheorem{conjecture}[theorem]{Conjecture}
\newtheorem{definition}[theorem]{Definition}

\DeclareMathOperator{\tr}{tr}
\DeclareMathOperator{\vectorize}{vec}
\DeclareMathOperator{\choi}{choi}

\title{\large\bfseries Apparent Universal Behavior in $2^{\text{nd}}$ Moments of Random Quantum Circuits \vspace{-0.2cm}}

\author[1]{Daniel Belkin} 
\author[1,2]{James Allen} 
\author[1]{Bryan K. Clark}

\affil[1]{\footnotesize Institute for Condensed Matter Theory and IQUIST and NCSA Center for Artificial Intelligence Innovation and Department of Physics, University of Illinois at Urbana-Champaign, IL 61801, USA}

\affil[2]{\footnotesize D\'epartement de Physique, Universit\'e de Montr\'eal, Montr\'eal, QC, Canada H3C 3J7}

\date{}
\begin{document}

\maketitle

\begin{abstract}
    Just how fast does the brickwork circuit form an approximate 2-design? 
    Is there any difference between anticoncentration and being a 2-design? 
    Does geometry matter?   
    How deep a circuit will I need in practice?
    We tell you everything you always wanted to know about second moments of random quantum circuits, but were too afraid to compute. Our answers generally take the form of numerical results for up to 50 qubits.
    
    Our first contribution is a strategy to determine explicitly the optimal experiment which distinguishes any given ensemble from the Haar measure. With this formula and some computational tricks, we are able to compute $t = 2$ multiplicative errors exactly out to modest system sizes. As expected, we see that most families of circuits form $\epsilon$-approximate $2$-designs in depth proportional to $\log n$. For the 1D brickwork, we work out the leading-order constants explicitly. Our semi-empirical formula for the approximate $2$-design depth takes the form
    $$ d_\text{brickwork} \approx\frac{\log \left(\frac{3}{ \pi^2}\frac{n}{\epsilon}\right)}{\log \frac{5}{4}} + O\left(\frac{1}{n^2}\right)$$
    For graph-sampled architectures, we find some exceptions which are much slower, proving that they require at least $\Omega(n)$ gates per site. This answers a question asked by ref. \cite{mittal_local_2023} in the negative. We explain these  exceptional architectures in terms of connectedness, corresponding loosely to a separation of timescales. Based on this intuition we conjecture universal upper and lower bounds for graph-sampled circuit ensembles. 
    
    For many architectures, the optimal experiment which determines the multiplicative error corresponds exactly to the collision probability (i.e. anticoncentration). However, we find that the star graph anticoncentrates much faster than it forms an $\epsilon$-approximate $2$-design. Finally, we show that one needs only ten to twenty layers to construct an approximate $2$-design for realistic parameter ranges. This is a large constant-factor improvement over previous constructions. We show that the parallel complete-graph architecture is not quite the fastest scrambler, partially resolving a question raised by ref. \cite{dalzell_random_2022}.
\end{abstract}

\section{Introduction}
The convergence of random circuit ensembles to approximate unitary designs has been the subject of much study. Areas of interest include both designing ensembles which give approximate unitary designs especially quickly\cite{cui_random_2025, suzuki_more_2025, schuster_random_2025, laracuente_approximate_2024} and determining the rate at which simpler or more generic random circuit architectures approximate global random unitaries.\cite{ambainis_quantum_2007, brandao_local_2016, harrow_approximate_2023, belkin_approximate_2024, mittal_local_2023, chen_incompressibility_2024} The rate of convergence is a natural and useful property to understand, since it controls the large-depth behavior of all other experimentally-observable properties. 

Certain structured architectures are known to form $\epsilon$-approximate $t$-designs in depth $\Theta(\log n)$, where $n$ is the number of sites. With even more structure, including ancilla qubits and non-Haar-random local operations, one can show that a design is formed in depth $\Theta(\log \log n)$. On the other hand, for more generic regularly-connected arrangements of Haar-random local gates, all that is known is that the $\epsilon$-approximate $t$-design depth lies somewhere between $\Omega(\log n)$ and $O(n)$ \cite{dalzell_random_2022, belkin_approximate_2024}.

Our first contribution is a reduction of the $t = 2$ multiplicative error to a relatively tractable mathematical formula.  While previous work has computed 
 upper or lower bounds on the multiplicative error from other circuit properties (such as the spectral gap up to $t = 6$) those bounds are believed to be very loose.\cite{brandao_local_2016, allen_conditional_2025}.  This formula allows the multiplicative error to be computed exactly for many common architectures in time $4^n$ (or even $2^n$ when certain symmetries are present). This is a dramatic improvement over the $64^n$ runtime encountered by a naive strategy. This makes this approach, to our knowledge, the first usable algorithm for evaluating multiplicative error.  Our second contribution is a set of numerical results obtained with this algorithm. 
 
Section \ref{sec:graphs} covers architectures involving Haar-random 2-qubit gates in locations sampled uniformly from the edges of some fixed graph over the sites. This class of architectures has been studied extensively, with bounds ranging from circuit size $O(n^2)$ (i.e. depth $O(n)$) for graphs with convenient structure to $O(n^9 \log n)$ for arbitrary graphs \cite{chen_incompressibility_2024, ambainis_quantum_2007, brandao_local_2016, mittal_local_2023, oszmaniec_epsilon-nets_2022}. We focus on two key open questions: 
First, do any graphs form $2$-designs in sublinear depth? Yes. We find empirically that several typical families of graphs appear to require $\Theta(\log n)$. Second, as posed by ref. \cite{mittal_local_2023}: Does every choice of graph form a $2$-design at the same asymptotic rate? No. We give families of graphs which can be proven to require depth at least $\Omega(n)$. However, motivated by ref.~\cite{belkin_approximate_2024}, we show that these counterexamples are in a certain sense poorly connected. Indeed, all graphs we examine require \textit{at least $\Omega(\log n)$ gates} and \textit{at most $O(\log n)$ connections} in order to form approximate $2$-designs. We further conjecture that the complete and linear graphs are extremal on these respective measures (illustrated in Figure \ref{fig:graph_conjectures}). 

Section \ref{sec:brickwork} discusses the 1D brickwork architecture. The architectures of refs \cite{schuster_random_2025, laracuente_approximate_2024} which are known to scramble in depth $\Theta(\log n)$ are ``censored brickworks'', i.e. the 1D brickwork with certain random gates fixed to the identity. It seems intuitive that adding additional random unitaries to the middle of a circuit shouldn't usually make the circuit further from the Haar measure, and so one expects the 1D brickwork to also form a 2-design in depth at most $O(\log n)$. This intuition, however, is known to be false in at least some cases \cite{belkin_absence_2025}. Nonetheless, we find that the brickwork behaves as expected. In particular, we give a semi-empirical formula with leading behavior
\begin{gather}
    d_{\textrm{brickwork}} \sim \frac{\log \left(\frac{3}{ \pi^2}\frac{n}{\epsilon}\right)}{\log \frac{5}{4}}
\end{gather}
This formula is in practice quite close to the true behavior (see Figure \ref{fig:brickwork_nscaling}). 

Another important open question is which architectures scramble fastest in practice. Ref. \cite{dalzell_random_2022} suggested that the architecture we term the Parallel Complete Graph (see Section \ref{sec:fast_definitions}) might be the ``fastest anticoncentrator.'' We show that although it is not quite the fastest scrambler, it is much faster than any graph-sampled architecture. We give a construction of an architecture which we can show forms an $0.01$-approximate $2$-design on 50 qubits with only 12 layers. The depth needed looks nearly independent of qubit count over numerically accessible sizes. We suggest other constructions which seem likely to scramble even faster. 

Finally, we discuss the relationship between anticoncentration and approximate $2$-design-ness. Our Theorem \ref{thm:multerr_expdef} gives a convenient conceptual connection between the two. Anticoncentration is essentially a weaker form of convergence to the Haar measure. It is known that certain circuit architectures anticoncentrate in depth $\Theta(\log n)$, and it's also known that other similar architectures form approximate $2$-designs in depth $\Theta(\log n)$, which suggests that anticoncentration might be only slightly weaker than approximate $2$-design-ness in practice. Indeed, ref.~\cite{heinrich_anti-concentration_2026} shows that for \textit{state} $2$-designs, the two are closely related. We show that the case of \textit{unitary} designs is somewhat more complicated. We find that the two measures of convergence are exactly equal in many cases, but we give exceptional examples in which they differ dramatically. 

\section{Theory}
\subsection{Basics}
Suppose we have $n$ sites, each with a local Hilbert space of dimension $q$. We have some distribution $\varepsilon$ over the unitary group. The $2$\textsuperscript{nd} moment operator of this distribution is a quantum channel given by 
\begin{gather}
        \Phi_\varepsilon(\rho) = \mathbb{E}_{U \sim \varepsilon}\left[(U^\dagger \otimes U^\dagger) \! \rho (U \otimes U)\right]
\end{gather}
We will also make use of the vectorization map, under which
\begin{gather}
        \vectorize(\Phi_\varepsilon) = \mathbb{E}_{U \sim \varepsilon}\left[(U^* \otimes U^* \otimes U \otimes U)\right]
\end{gather}

The \textbf{multiplicative error} $\mathcal{M}(A,B)$ of a channel $A$ relative to a second channel $B$ is defined to be the smallest \(\epsilon\) such that $(1 + \epsilon) B - A$ and $A - (1 - \epsilon) B$ are both completely positive maps \cite{chen_incompressibility_2024}. 
Equivalently, consider applying the channel $A$ to a state $\rho$ and measuring a projector $\Pi$ which accepts with probability $\tr \left(\Pi A(\rho)\right)$, and likewise for $B$. Then we may write
\begin{gather}
    \label{eq:multerr_expdef}
    \mathcal{M}(A,B) = \max_{\rho, \Pi} \left |\frac{\tr \left(\Pi \left[A \otimes I\right](\rho)\right)}{\tr \left(\Pi \left[B \otimes I\right](\rho) \right)} - 1\right |
\end{gather}
In other words, the multiplicative error is a statement about the \textit{best-case experiment} for distinguishing the two channels. The ratio of Eq \ref{eq:multerr_expdef} is precisely the largest likelihood ratio obtainable from any single observed event. 
An \textbf{$\epsilon$-approximate 2-design} is an ensemble $\varepsilon$ whose $2$\textsuperscript{nd} moment channel $\Phi_\varepsilon$ has a multiplicative error of at most $\epsilon$ with that of the Haar measure over the global Hilbert space, i.e. 
\begin{gather}
    \mathcal{M}(\Phi_\varepsilon, \Phi_\text{Haar}) \leq \epsilon
\end{gather}
Approximate designs are also often defined in terms of other error metrics, but we will focus on the multiplicative error here. 

\subsection{Constraining the Optimal Experiment}
We will require that the ensemble $\mathcal{E}$ 
\begin{itemize}
    \item Is invariant under the action of single-site unitaries (\textbf{local invariance})
    \item Gives rise to a 2\textsuperscript{nd} moment operator whose vectorization is positive-semidefinite. (\textbf{PSD vectorization})
\end{itemize}

The second condition is a bit tricky to interpret. However, it can be shown to hold for many ensembles of interest, such as graph-sampled circuits or the 1D brickwork at odd depths (see Appendix~\ref{app:psd_circuits}). Furthermore, given a locally invariant ensemble $\mathcal{E}$, one may define an ensemble $\mathcal{E}'$ with a PSD vectorization by sampling $U V^\dagger$, with $U,V$ drawn i.i.d. from $\mathcal{E}$. 

\begin{theorem}
    \label{thm:multerr_expdef}
    Let $\Phi_\varepsilon$ be the $2$nd moment operator of a locally invariant distribution over $\mathcal{U}\left(q^n\right)$ with a PSD vectorization. Define 
    \begin{gather}
    \ket{\psi(x)} =  \begin{cases} \ket{0 0}  & x = 0 \\
   \frac{1}{\sqrt{2}} \left(\ket{01} - \ket{10}\right) & x = 1
    \end{cases}
    \end{gather}
    The multiplicative error between $\Phi_\varepsilon$ and the 2nd moment operator $\Phi_{\text{Haar}}$ of the Haar distribution over $\mathcal{U}\left(q^n\right)$, as given by the maximization in Equation \ref{eq:multerr_expdef}, is saturated by the choice 
    \begin{gather}
        \rho = \Pi = \bigotimes_i \ket{\psi(a_i)}\bra{\psi(a_i)}
    \end{gather}
    for some $\vec{a} \in \{0,1\}^{n}$. 
    In other words, 
    \begin{gather}
        \label{eq:experimental_multerr_hermitian}
        \mathcal{M}\left(\Phi_\varepsilon, \Phi_\text{Haar}\right) = \max_{\vec{a} \in \{0,1\}^{n}} \frac{\tr \left[\rho_{\vec{a}} \Phi_\varepsilon\left(\rho_{\vec{a}}\right)\right]}{\tr \left[\rho_{\vec{a}} \Phi_\text{Haar}\left(\rho_{\vec{a}}\right)\right]} - 1
    \end{gather}
\end{theorem}
A proof is given in Appendix~\ref{app:multiplicative_errors}. This theorem replaces the maximum over all possible experiments in Equation \ref{eq:multerr_expdef} with a maximum over a finite set of possibilities. Note that the collision probability, often used to define anticoncentration \cite{dalzell_random_2022}, corresponds to the choice $\vec{a} = \vec{0}$. This relationship is discussed in more detail in Section \ref{sec:anticoncentration}.

\subsection{Explicit Form for Numerics}
Consider the $t$\textsuperscript{th} moment of an ensemble $\mathcal{E}$ on $n$ sites of local Hilbert space dimension $q$.  Local invariance of $\mathcal{E}$ implies that the vectorized moment operator $\vectorize \Phi_{\mathcal{E}}$ involves a projection into the commutant of $\mathcal{U}(q)$ on each site. This fact makes our computations a bit simpler. By Schur-Weyl duality, the commutant is spanned by states labeled by permutations. We term these \textbf{permutation basis states}, explicitly
\begin{gather}
\ket{\sigma} = \frac{1}{\sqrt{q}^t} \sum_{\vec{i} \in \{1...q\}^t} \ket{\vec{i}} \otimes \ket{\sigma(\vec{i})}
\end{gather}
Here the permutation acts by permuting the order of the elements of $\vec{i}$ \cite{allen_conditional_2025}. For $t = 2$ the only permutations are identity and swap, so the dimension of the local commutant is always $2$. To obtain our numerical results, we express the moment operator in this basis. This corresponds to finding coefficients $H_{\vec{\sigma}, \vec{\tau}}$ such that
\begin{gather}
    \vectorize\left(\Phi_\varepsilon - \Phi_\text{Haar}\right)\ket{\sigma_1...\sigma_n} = \sum_{\tau_1...\tau_n} H_{\tau_1...\tau_n, \sigma_1...\sigma_n}\ket{\tau_1...\tau_n}
\end{gather}
If one then defines 
\[\mathbf{v}(\vec{a}) = \bigotimes_i \begin{bmatrix} 1 \\ (-1)^{a_i} \end{bmatrix}\]
we end up with 
\begin{align}
    \mathcal{M}\left(\Phi_\varepsilon, \Phi_\text{Haar}\right) &= \frac{1}{2} \frac{1}{\left(1 + \frac{1}{q}\right)^n} 
\max_{\vec{a} \in \{0,1\}^{n}} \left[\left(1 + (-1)^{\sum_i a_i}q^{-n} \right)
\left(\frac{1 + \frac{1}{q}}{1 - \frac{1}{q}}\right)^{\sum_i a_i}
\mathbf{v}(\vec{a})^T H \mathbf{v}(\vec{a})
\right]
\\& \sim 
\left(\frac{2}{3}\right)^n 
\max_{\vec{a} \in \{0,1\}^{n}} \left[3^{\sum_i a_i}
 \mathbf{v}(\vec{a})^T H \mathbf{v}(\vec{a}) 
\right]
\end{align}
where in the second line we've taken $q = 2$ and dropped $e^{-O(n)}$ contributions to emphasize the key structure of the formula. 

In this work we evaluate $\mathbf{v}(\vec{a})^T H \mathbf{v}(\vec{a})$ exactly using tensor network methods \cite{braccia_computing_2024}. It is also possible in principle to approximate Equation \ref{eq:experimental_multerr_hermitian} more directly by sampling random Clifford gates, since the Clifford group is a 2-design. However, in practice we have found that this converges poorly. Circuits composed of Haar-random local unitaries tend to self-average quite well, such that only a small number of samples are needed to estimate $\tr \left[\rho_{\vec{a}} \Phi_\varepsilon\left(\rho_{\vec{a}}\right)\right]$. Although Clifford circuits have the same second moments, the more discrete distributions require a much larger number of samples for averages to converge. Another approach is to instead contract the tensor network approximately, e.g. with belief propagation, by exploiting positivity bias, or by mapping to a stat mech model and running Monte Carlo simulations. While these approaches would be useful for larger numbers of qubits, exact contraction is adequate to address the questions at hand here.

\section{Graphs}
\label{sec:graphs}
We now consider graph-sampled architectures, in which a Haar-random 2-site unitary is applied to a random pair of sites chosen from the uniform distribution over the edges of some specified connectivity graph. The moment operator is averaged over both the possible circuit structures and the possible values of each local unitary. Our core goal in this section is to understand which structural properties of a graph cause especially fast or slow scrambling.

\subsection{Prior Work}
Ref.~\cite{ambainis_quantum_2007} established an approximate \(2\)-design size of at most \(O(n^2)\) gates for the complete graph. Ref.~\cite{brandao_local_2016} found \(O(n^2)\) for the linear graph. Ref.~\cite{oszmaniec_epsilon-nets_2022} used a similar strategy for graphs which admit a Hamiltonian path to obtain a bound scaling as \(O(n^3)\). Ref.~\cite{mittal_local_2023} considers graphs with \(|E|\) edges, bounded degree, and bounded effective spanning-tree height, obtaining an \(O(|E|n)\) bound for this case. Without any structural assumptions, the best available bound comes from the results of ref.~\cite{chen_incompressibility_2024}, which imply a bound of 
\[O\left(n^9 \log n\right)\]
gates for arbitrary graphs. Ref. \cite{belkin_approximate_2024}, meanwhile, proves that the approximate $t$-design depth is related to the number of connected blocks into which the typical realization can be divided. It also proposes conjectures based on this approach which would imply an $O(n^3 \log n)$ bound. 

\subsection{Error vs. Circuit Size}
Figure \ref{fig:err_vs_gates} shows multiplicative error vs. circuit size for linear, circle, complete, and lollipop graphs (see Figure \ref{fig:slow_graphs}b for the definition of the lollipop graph). In all figures we choose local dimension $q = 2$. 

\begin{figure}[H]
    \centering
    \begin{tikzpicture}
        \begin{scope}
            \node[anchor=north west,inner sep=0] (image_a) at (0,0)
            {\includegraphics[width=0.6\columnwidth]{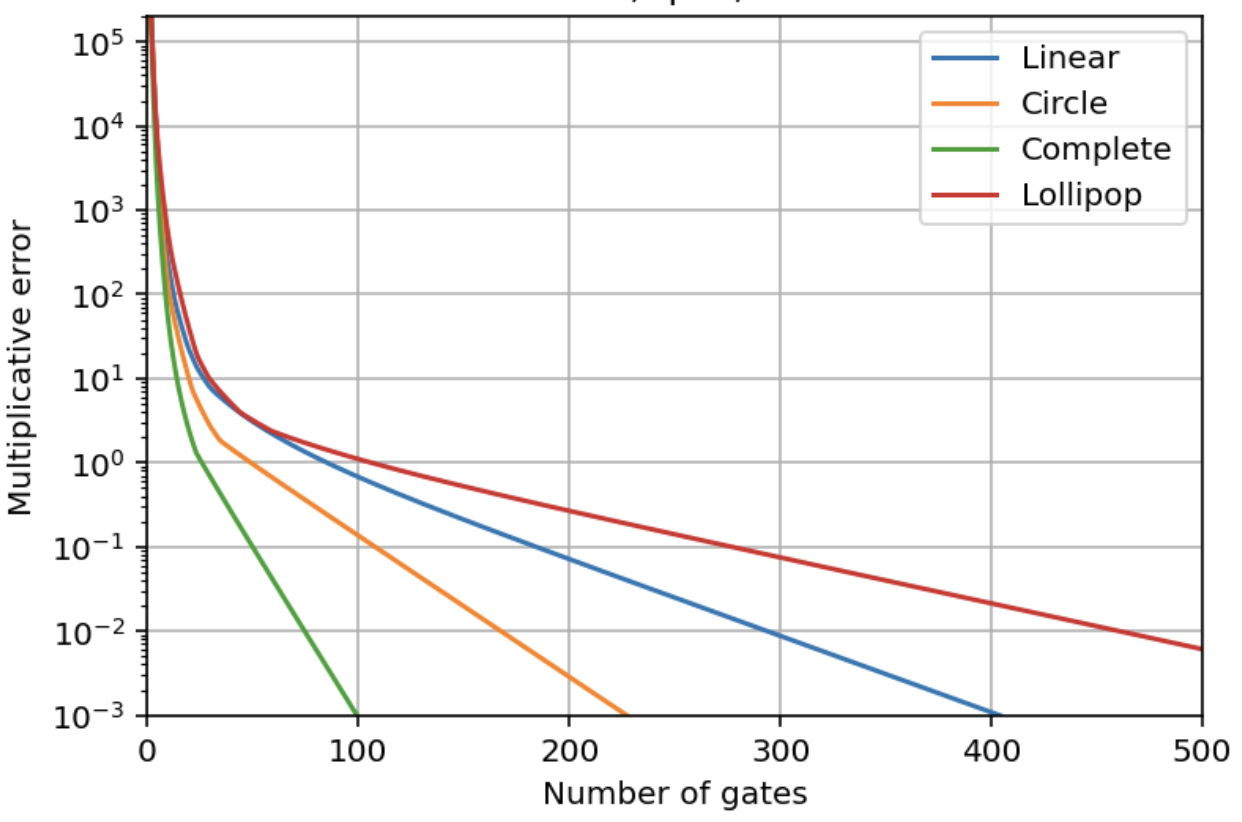}};
        \end{scope}
    \end{tikzpicture}
    \vspace*{-0.4cm}
    \caption{Multiplicative error $\mathcal{M}$ at $t = 2$ vs. gate count $s$ for four graphs on 12 qubits.}
    \label{fig:err_vs_gates}
\end{figure}

These curves share a common structure: A very rapid initial drop in $\epsilon$, followed by a uniform exponential decay. To understand this, let $\lambda_i$ be the (unique, ordered) eigenvalues of the vectorized single-step moment operator $\vectorize \left(\Phi_\varepsilon\right)$ and let $P_i$ project into the corresponding eigenspaces. Then at circuit size $s$, we may write
\begin{gather} 
    \mathcal{M} = \max_{\vec{a}} \sum_{i > 0} \frac{\left|\left|P_i \vectorize (\rho_{\vec{a}})\right|\right|^2}{\left|\left|P_0 \vectorize (\rho_{\vec{a}})\right|\right|^2}\lambda_i^s \label{eq:mult_error_over_depth}
\end{gather}
(see Appendix~\ref{app:scaling_depth} for details). At large depths only the dominant eigenvalue $\lambda_1$ matters, which contributes the straight lines to Figure \ref{fig:err_vs_gates}. Each of these lines is of the form
\begin{gather}
    \log \mathcal{M} \approx s \log \lambda_1 
    + \max_{\vec{a}} \log
    \frac{\left|\left|P_1 \vectorize (\rho_{\vec{a}})\right|\right|^2}{\left|\left|P_0 \vectorize (\rho_{\vec{a}})\right|\right|^2}
\end{gather}
In other words, the small-$\epsilon$ behavior is determined by the norm of the projection of the optimal $\vectorize \rho_{\vec{a}}$ into the dominant eigenspace. Early work on approximate $t$-designs was based on determining the multiplicative error using only the spectral gap, which corresponds to assuming that the dominant experiment $\vectorize \left(\Phi_\varepsilon\right)$ lies entirely in the dominant eigenspace: 
\begin{gather}
    \log \mathcal{M} \leq s \log \lambda_1 + \max_{\vec{a}} \log\frac{\left|\left|\vectorize (\rho_{\vec{a}})\right|\right|^2}{\left|\left|P_0 \vectorize (\rho_{\vec{a}})\right|\right|^2}
\end{gather}
In practice this estimate seems to be quite loose. This is the same as approximating the curves in Figure \ref{fig:err_vs_gates} as straight lines, with the initial values and final slopes unchanged but without the ``elbows'' on the left side of the plot. The rapid early drop corresponds to subdominant eigenspaces. The drops are large and fast, which indicates that most of the norm of the dominant irrep lies in eigenspaces with eigenvalues much smaller than $\lambda_1$.

\subsection{Critical Depth vs. Qubit Count}
\label{sec:graphs_by_nqubit}
Figure \ref{fig:basic_depths} shows the circuit size needed to reach multiplicative error $0.01$ for various graphs and system sizes. 
\begin{figure}[H]
    \centering
    \begin{tikzpicture}
        \begin{scope}
            \node[anchor=north west,inner sep=0] (image_a) at (0,0)
            {\includegraphics[width=0.6\columnwidth]{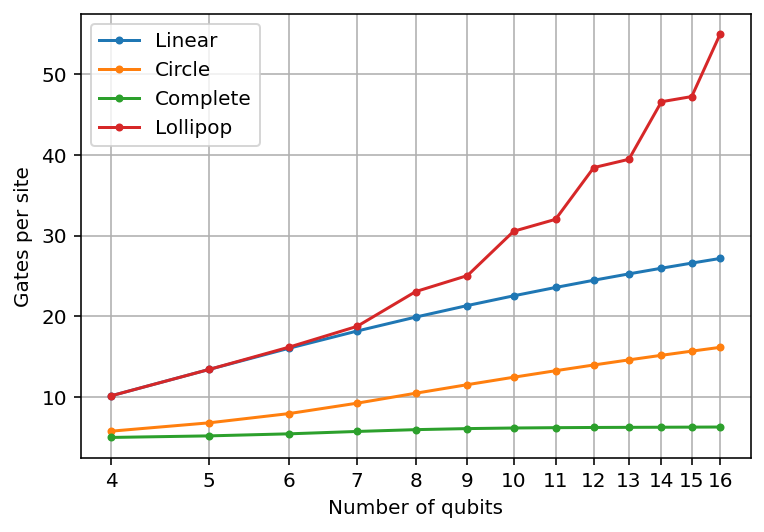}};
        \end{scope}
    \end{tikzpicture}
    \vspace*{-0.4cm}
    \caption{Circuit size needed to reach an $0.01$-approximate $2$-design for linear, circle, complete, and lollipop graphs.}
    \label{fig:basic_depths}
\end{figure}
We see that the linear and circle graphs both give roughly straight lines on this plot, which is to say they form approximate $2$-designs in depth\footnote{Strictly speaking the depth of graph-sampled architectures is not well-defined since it depends on the realization. Here we presume that the typical ``depth'' is proportional to the number of gates per site.} $O(\log n)$. On the other hand, the complete graph appears nearly flat in comparison. There is a lower bound of depth $\Omega(\log n)$ for the complete graph via anticoncentration\cite{dalzell_random_2022}, but this behavior is difficult to discern from numerically-accessible system sizes.\footnote{Fig. \ref{fig:fast} shows more data for the complete graph and compares it against the lower bound of Ref.~\cite{dalzell_random_2022}.} The lollipop line curves upwards. This suggests strongly that not all graphs scale at the same asymptotic rate. This is discussed in detail in Section~\ref{sec:lollipop} below. 

Figure \ref{fig:fancy_depths} gives analogous curves for some other families of graphs. Generally we see more dense graphs tend to form approximate $2$-designs faster, with both trees and Ramanujan graphs appearing to interpolate between the linear and complete cases as the degree of the nodes increases. The lollipop is our only exception to this trend. 

\begin{figure}[H]
    \centering
    \begin{tikzpicture}
        \begin{scope}
            \node[anchor=north west,inner sep=0] (image_a) at (0,0)
            {\includegraphics[width=0.4\columnwidth]{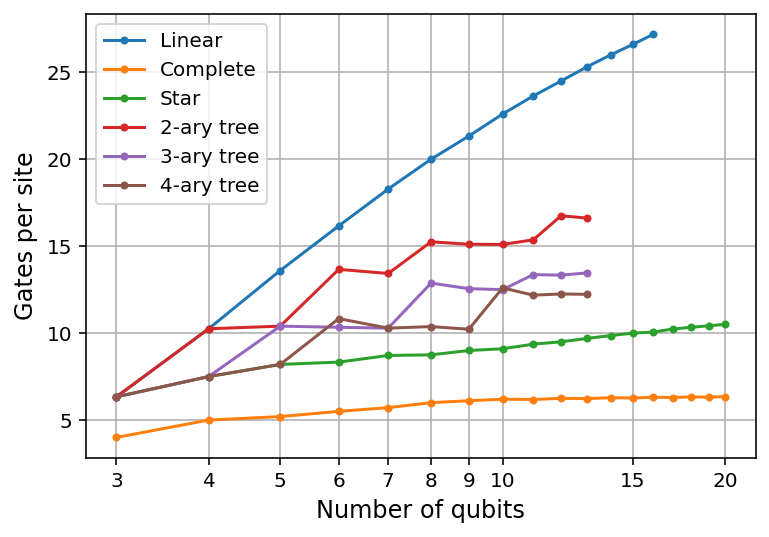}};
        \end{scope}
        \node [anchor=north west] (note) at (-0.2,0) {\small{\textbf{a)}}};
        \begin{scope}
            \node[anchor=north west,inner sep=0] (image_b) at (7.5,0)
            {\includegraphics[width=0.4\columnwidth]{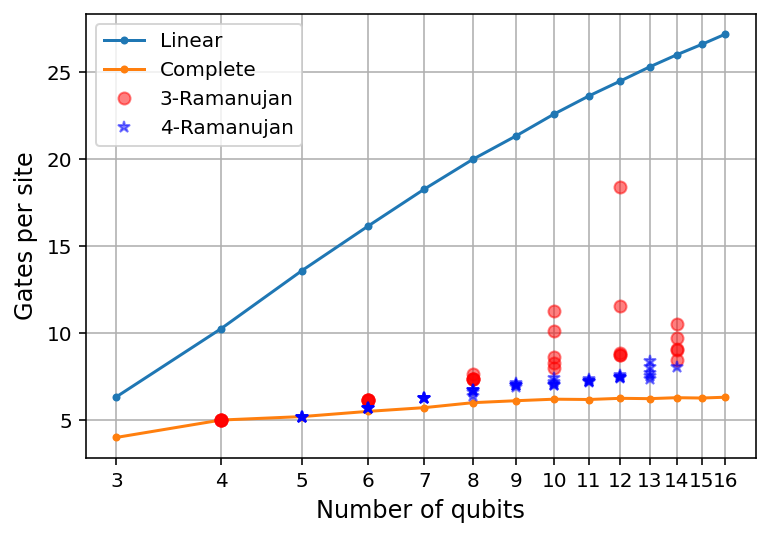}};
        \end{scope}
        \node [anchor=north west] (note) at (7.3,0) {\small{\textbf{b)}}};
    \end{tikzpicture}
    \vspace*{-0.4cm}
    \caption{Circuit size needed to reach an $0.01$-approximate $2$-design for some other families of graphs. Results for complete and linear graph are repeated for reference. (a) Tree and star graphs. (b) Several random $d$-regular Ramanujan graphs.}
    \label{fig:fancy_depths}
\end{figure}

\subsection{Connectedness}
\begin{figure}[h]
    \centering
    \begin{tikzpicture}
        \begin{scope}
            \node[anchor=north west,inner sep=0] (image_a) at (0,0)
            {\includegraphics[width=0.45\columnwidth]{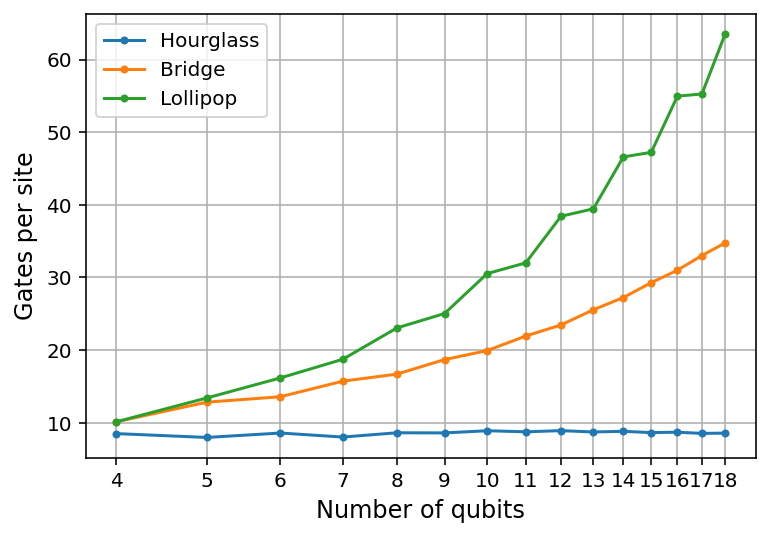}};
        \end{scope}
        \node [anchor=north west] (note) at (-0.1,-0.3) {\small{\textbf{a)}}};
        \begin{scope}
            \node[anchor=north west,inner sep=0] (image_b) at (8.9,-0.33)
            {\includegraphics[width=0.35\columnwidth]{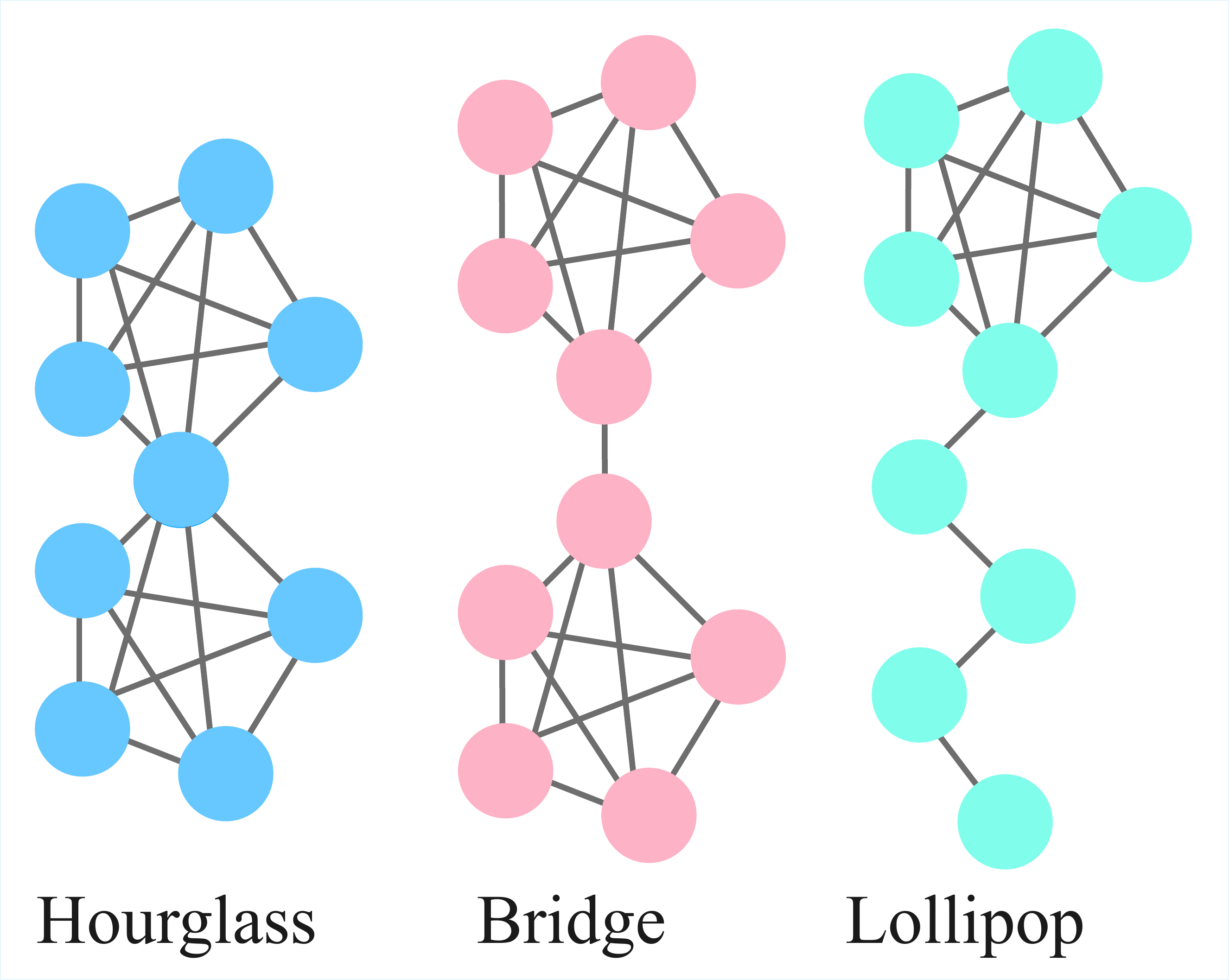}};
        \end{scope}
        \node [anchor=north west] (note) at (9,-0.3) {\small{\textbf{b)}}};
    \end{tikzpicture}
    \vspace*{-0.4cm}
    \caption{(a) $0.01$-approximate $2$-design depths for each of three families of graphs. Although the hourglass and bridge look very similar, their scrambling rates are very different. (b) Illustrations of the hourglass, bridge, and lollipop graph families. In each case we assign $\lceil \frac{n}{2} \rceil$ nodes to the upper clique, such that the two regions are of roughly equal size.}
    \label{fig:slow_graphs}
\end{figure}

\subsubsection{Why is the lollipop special?} \label{sec:lollipop}
We saw in Figures \ref{fig:basic_depths} and \ref{fig:fancy_depths} that all of these architectures lie somewhere between the linear and the complete graph except for one. The lollipop graph is not only much slower to scramble than the other architectures shown, this gap increases rapidly with $n$. To understand this behavior, recall that we are choosing gate locations uniformly from all the edges of the graph. The lollipop has ${n/2 \choose 2} = O(n^2)$ edges in the ``candy'', but only $\frac{n}{2}$ edges in the ``stick''. It follows that the vast majority of random gates we draw will act in the candy, with only a fraction $O\left(\frac{1}{n}\right)$ helping to scramble the stick. The stick resembles a linear graph. As we saw above, the linear graph requires $O(n \log n)$ gates to scramble, and so we should expect the lollipop to require $O(n^2 \log n)$ gates before the stick becomes well-scrambled. 

To build some more intuition, consider two other families of graphs. The \textbf{hourglass graph} is two cliques which share a single node. The \textbf{bridge graph} is two cliques connected by a single edge. These two geometries are extremely similar to each other, as illustrated in Figure \ref{fig:slow_graphs}b. And yet we see in Figure \ref{fig:slow_graphs}a that these two very similar architectures have radically different scrambling speeds. Why? In the bridge architecture, information can scramble very well within each clique. But the rate of scrambling \textit{between} cliques is bottlenecked by the bridge itself, which occurs only once every $O\left(\frac{1}{n^2}\right)$ gates. This behavior is illustrated in Figure \ref{fig:connections_slow}a. The hourglass has no such bottleneck. This explains the large difference in scrambling rates between two otherwise similar architectures. 

\begin{figure}[H]
    \centering
    \begin{tikzpicture}
        \begin{scope}
            \node[anchor=north west,inner sep=0] (image_b) at (0,0)
            {\includegraphics[width=0.6\columnwidth]{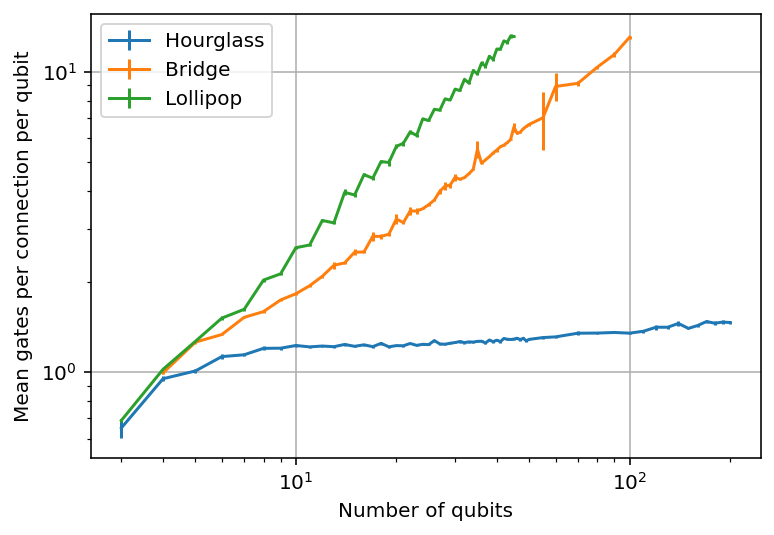}};
        \end{scope}
    \end{tikzpicture}
    \vspace*{-0.4cm}
    \caption{Mean gates per connection per qubit for each of three architectures, as estimated by the greedy algorithm described in Appendix \ref{app:connections_of_architectures}. We see that the hourglass is regularly-connected even at large $n$, while the bridge and lollipop become poorly connected as $n$ grows.}
    \label{fig:connections_slow}
\end{figure}

There is a physical interpretation for this behavior. We can think of the action of random local unitaries as being similar in spirit to the evolution of a physical system under a ``generic'' (e.g. chaotic) local Hamiltonian. The behavior of the lollipop is just a separation of timescales: The head experiences strong interactions and thermalizes quickly, while the tail experiences only very weak interactions and so thermalizes very slowly. Similarly, the two ``islands'' of the bridge graph are quite quick to thermalize internally, but the exchange of quantum information between the two is very slow. This resembles prethermalization of two weakly-interacting subsystems to independent temperatures. 

Motivated by ref.~\cite{belkin_approximate_2024}, we suggest a unifying description of the behavior of the lollipop and bridge. Theorem 3 of that work establishes a bound on the spectral gaps of random circuits in terms of the number of connected blocks into which they can be divided. The hourglass is connected after $\Theta(n \log n)$ gates, while the bridge requires $\Theta(n^2)$ gates and the lollipop $\Theta(n^2 \log n)$ gates.\footnote{Each clique is of size $\frac{n}{2}$, so by percolation are connected after $\Theta(n \log n)$ gates. For the hourglass this is sufficient to ensure the whole graph is connected. For the bridge graph, however, we also need the bridge itself to be sampled, which occurs only once in every $2 {n/2 \choose 2} + 1$ gates. For the lollipop, we have percolation in the candy in $\Theta(n \log n)$, but the coupon collector problem in the stick requires $\Theta(n^2 \log n)$ gates.} These asymptotics suggest an explanation for the differences seen in Figure \ref{fig:slow_graphs}a. 

In fact, this intuition can be formalized. 
\begin{theorem}
    \label{thm:disconnected}
    The bridge graph on $n$ sites requires at least 
    $s \geq \frac{n(n-2)}{4} \log \frac{1}{\epsilon}$
    gates in order to form a multiplicative-error $\epsilon$-approximate $t$-design. 
\end{theorem}
A proof is given in Appendix \ref{app:disconnected}. Ref. \cite{mittal_local_2023} asked if there is a universal asymptotic form for the circuit size needed for a graph-sampled architecture to give an approximate $t$-design. Together with the numerics shown in Figure \ref{fig:basic_depths}, this theorem strongly suggests that the answer is no. On the other hand, the exceptions we exhibit are due only to poor connectivity, which is somewhat trivial. There remain, then, two questions: Is failure-to-connect the only way to evade fast scrambling? Can we salvage any universal characterization of the scrambling rates of graph-sampled architectures?

\subsubsection{Results by connection count}
Figure \ref{fig:connections_basic} shows mean connection count by circuit size for several graphs, as estimated by the greedy algorithm described in Appendix~\ref{app:connections_of_architectures}.

\begin{figure}[H]
    \centering
    \begin{tikzpicture}
        \begin{scope}
            \node[anchor=north west,inner sep=0] (image_a) at (0,0)
            {\includegraphics[width=0.6\columnwidth]{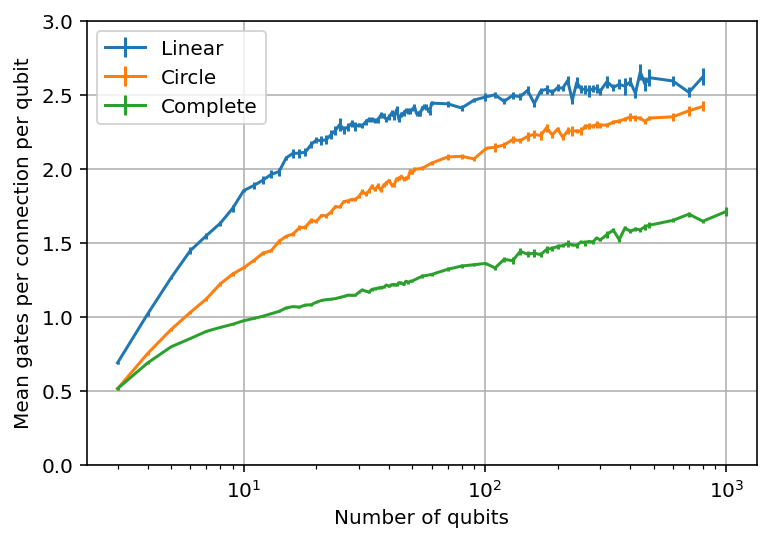}};
        \end{scope}
    \end{tikzpicture}
    \vspace*{-0.4cm}
    \caption{Mean connected blocks per gate per site, as estimated by the greedy algorithm described in Appendix \ref{app:connections_of_architectures}, for each of three families of graphs.}
    \label{fig:connections_basic}
\end{figure}

We can now repeat the approximate $2$-design depth calculations shown in Section \ref{sec:graphs_by_nqubit}, with the vertical axis rescaled to be in terms of connections counts. Results are shown in figures \ref{fig:basic_connections} and \ref{fig:fancy_connections}.

\begin{figure}[H]
    \centering
    \begin{tikzpicture}
        \begin{scope}
            \node[anchor=north west,inner sep=0] (image_b) at (0,0)
            {\includegraphics[width=0.6\columnwidth]{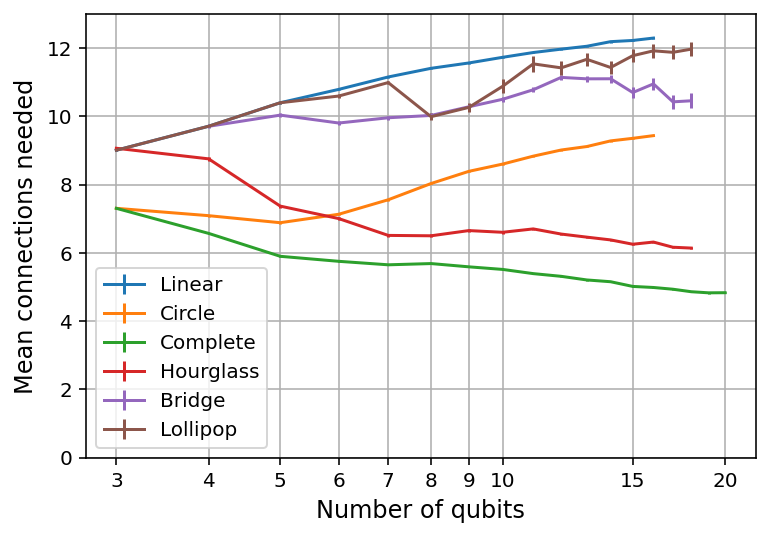}};
        \end{scope}
    \end{tikzpicture}
    \vspace*{-0.4cm}
    \caption{Connection count needed to reach an $0.01$-approximate $2$-design for each of six graph families. We see that all require roughly comparable connection counts, although some rise slightly with $n$ and others fall.}
    \label{fig:basic_connections}
\end{figure}

\begin{figure}[H]
    \centering
    \begin{tikzpicture}
        \begin{scope}
            \node[anchor=north west,inner sep=0] (image_a) at (0,0)
            {\includegraphics[width=0.4\columnwidth]{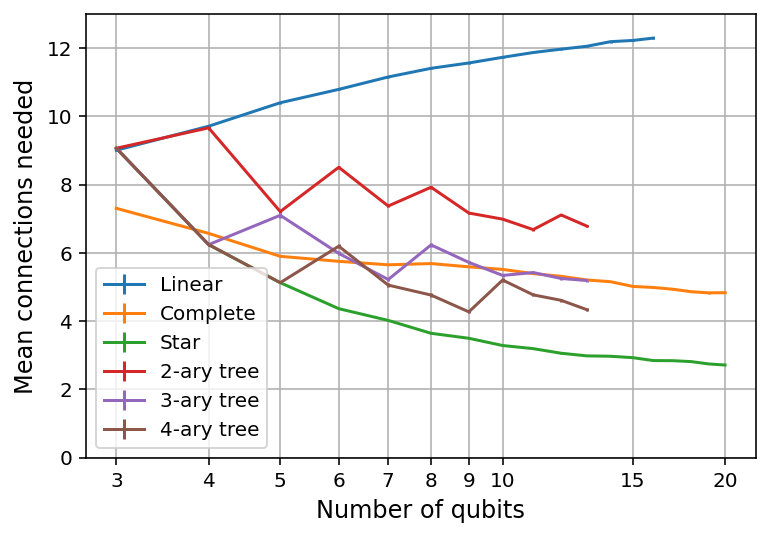}};
        \end{scope}
        \node [anchor=north west] (note) at (-0.2,0.1) {\small{\textbf{a)}}};
        \begin{scope}
            \node[anchor=north west,inner sep=0] (image_b) at (7.5,0)
            {\includegraphics[width=0.4\columnwidth]{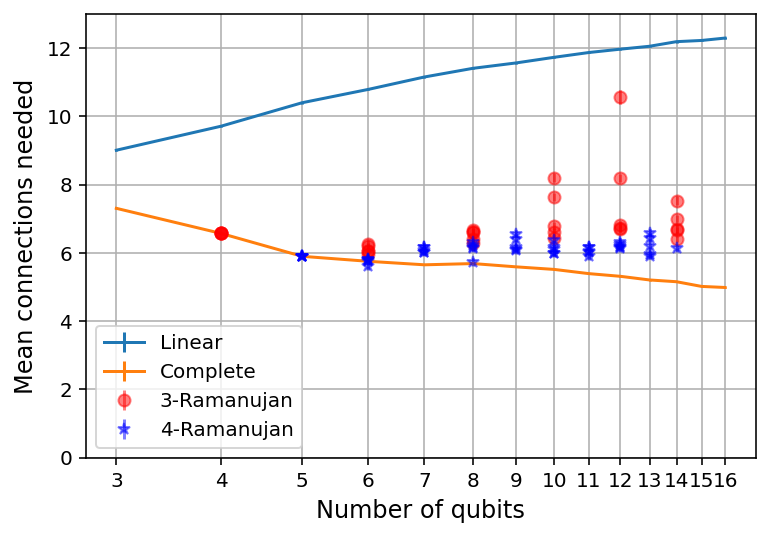}};
        \end{scope}
        \node [anchor=north west] (note) at (7.3,0.1) {\small{\textbf{b)}}};
    \end{tikzpicture}
    \vspace*{-0.4cm}
    \caption{Connection count needed to reach an $0.01$-approximate $2$-design for  (a) tree and star graphs, (b) Ramanujan graphs.}
    \label{fig:fancy_connections}
\end{figure}

In terms of connection count, the slowest-scrambling architecture tested is the linear graph, while the fastest is the star graph. Note that star graph has $O(n \log n)$ gates per connection, where as the linear graph need $O(n \log n)$ gates for the first connection but only $O(n)$ for later connections (see Appendix \ref{app:connections_of_architectures}). The lollipop graph looks quite similar to the brickwork, which is what one expects since the dominant contribution is due to the linear ``stick'' portion of the graph.

\subsection{Conjectures}
All of these results are consistent with two conjectures, illustrated in Figure \ref{fig:graph_conjectures}. We do not test every possible graph, nor $t > 2$, so the full strength of these is somewhat speculative. 

\begin{conjecture}
    \label{conj:gates}
    No other graph on $n$ qudits forms an $\epsilon$-approximate $t$-design with fewer gates than the complete graph, which requires $\Theta(n \log n)$ gates.
\end{conjecture}
\begin{conjecture}
    \label{conj:connections}
    No other graph on $n$ qudits requires more connections to form an $\epsilon$-approximate $t$-design than the linear graph, which requires $\Theta(\log n)$ connections.
\end{conjecture}
\begin{figure}[H]
    \centering
    \begin{tikzpicture}
        \begin{scope}
            \node[anchor=north west,inner sep=0] (image_a) at (0,0)
            {\includegraphics[width=0.4\columnwidth]{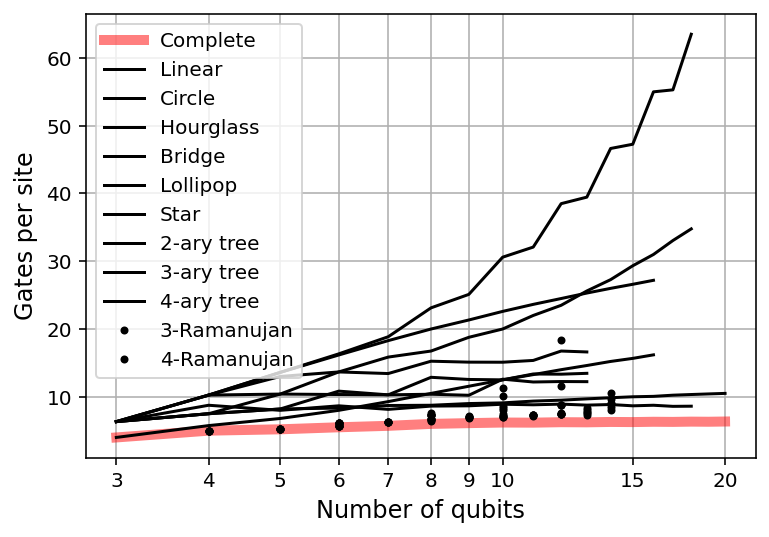}};
        \end{scope}
        \node [anchor=north west] (note) at (-0.2,0) {\small{\textbf{a)}}};
        \begin{scope}
            \node[anchor=north west,inner sep=0] (image_b) at (7.4,0)
            {\includegraphics[width=0.4\columnwidth]{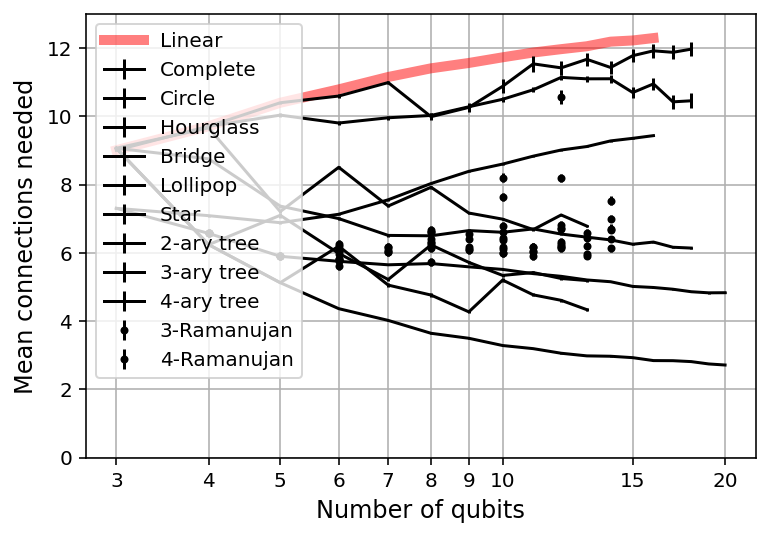}};
        \end{scope}
        \node [anchor=north west] (note) at (7.2,0) {\small{\textbf{b)}}};
    \end{tikzpicture}
    \vspace*{-0.4cm}
    \caption{Conjectured bounds. (a) Circuit size needed to reach an $0.01$-approximate $2$-design. The complete graph is fastest. (b) Connection counts needed to reach an $0.01$-approximate $2$-design. The linear graph is the slowest.}
    \label{fig:graph_conjectures}
\end{figure}

From Figure \ref{fig:graph_conjectures}b, it is not clear that the connection count needed by the linear graph scales as $\Theta(\log n)$. However, we argue in appendix \ref{app:connections_of_architectures} that this scaling is likely to emerge at much larger $n$. Similarly, the $\Theta(n \log n)$ scaling of the complete graph is difficult to discern from Figure \ref{fig:graph_conjectures}a, but we show in Figure \ref{fig:fast} that an $\Omega(n \log n)$ lower bound will become effective at much larger $n$.

The connection depth for any graph is at most $|E| \log n$.\footnote{Proof: Choose a spanning tree of $n-1$ edges. A fraction $(n-1)/|E|$ of the gates will land on that tree. The graph is connected once the coupon collector problem on the tree is solved, which requires $O(n \log n)$ tree edges, or $O(|E| \log n)$ total edges.} So conjecture \ref{conj:connections} also implies an upper bound of $O(|E|(\log n)^2)$ gates, which is at worst $O(n^2 (\log n)^2)$. It seems plausible that the lollipop graph may saturate this bound.

\section{Brickwork}
\label{sec:brickwork}
\subsection{Prior work}
The convergence of the 1D brickwork random circuit to the Haar measure has been the subject of much study\cite{brandao_local_2016,haferkamp_random_2022, haferkamp_improved_2021, chen_incompressibility_2024, allen_conditional_2025, dalzell_random_2022}. The first case to be understood was anticoncentration, which asks when the collision probabilities of computational basis measurements become similar. This was essentially resolved by ref. \cite{dalzell_random_2022}, which gave both upper and lower bounds scaling as $\frac{\log n}{\log \frac{q^2 + 1}{2q}}$. The rate of convergence of the collision probability to the Haar-measure is thus very well-understood. What remains open is whether or not all possible experiments behave similarly. 

Until 2024, it was widely assumed that there existed some observables which required depth $O(n)$ to converge. However, refs. \cite{schuster_random_2025} and \cite{laracuente_approximate_2024} proved that certain brickwork-like architectures form approximate $t$-designs in depth $O(\log n)$. More precisely, these architectures are 1D brickworks with certain gates removed. This suggested strongly that the scaling of the 1D brickwork approximate $t$-design depth would also be $O(\log n)$. Indeed, a conjecture of ref.~\cite{harrow_approximate_2023} implies that removing random gates from an architecture can never increase the distance from the Haar measure, which would have sufficed for a proof. Ref. \cite{belkin_absence_2025}, however, constructs a counterexample to that conjecture. The question of whether all observables converge at the same rate as the collision probability thus remains open. 

Here we provide convincing numerical evidence in the case $t = 2$. The best known upper bound in this case comes from  combining the exact spectral gap of \cite{deneris_exact_2024} with Theorem 54 of ref. \cite{allen_conditional_2025} to obtain
\begin{gather}
    d(n,q,\epsilon) \leq 1 + \frac{2 n t \log q + \log \frac{1}{\epsilon}}{\log \frac{q^2 + 1}{2q}} 
\end{gather}
(see also \cite{znidaric_solvable_2022}). 
The best known lower bound, on the other hand, is via ref. \cite{dalzell_random_2022}, which is roughly of the form 
\begin{gather}
     d(n,q,\epsilon) \geq \frac{\log n - 13.81}{\log \frac{q^2 + 1}{2q}}
\end{gather}
See Appendix~\ref{app:scaling_depth} for a more careful bound.  Note that this bound is for the case of periodic boundary conditions, although it seems likely that essentially the same argument goes through with open boundary conditions. 

\subsection{Results}
Figure \ref{fig:brickwork_nscaling} shows the 0.01-approximate $2$-design depth of the open-boundary-condition 1D brickwork in terms of qubit count. We find empirically that the optimal experiment always corresponds to preparing an antisymmetric state on the two endpoints of the line and a symmetric state in the bulk. In the language of Eq \ref{eq:experimental_multerr_hermitian} this is 
\[\vec{a} = \begin{bmatrix} 
1 & 0 & 0 & \dots & 0 & 0 & 1 
\end{bmatrix}\] 
We call this irrep the \textbf{entangled boundaries} experiment, since the two copies of the system are unentangled everywhere except for the edges. 

For this particular choice of experiment we compute the error numerically out to 50 qubits. Furthermore, applying our Equation \ref{eq:experimental_multerr_hermitian} to this $\vec{a}$ and using the work of ref. \cite{deneris_exact_2024} allows one to show, via a rather tedious calculation, that the dominant large-$n$, small-$\epsilon$ behavior of the 2-design depth is of the form 
\begin{align}
    \label{eq:brickwork_formula_template}
    f(n,q,\epsilon) &= \alpha \left(\log n - \log \epsilon\right) + \beta
    \\&= \frac{\log \left[\frac{2}{\pi^2}\frac{q^2 -1}{q} \frac{n}{\epsilon}\right]}{\log \frac{q^2 + 1}{2q}} + O\left(\frac{1}{n^2}\right)
\end{align}
with parameters
\begin{gather}
\label{eq:brickwork_alpha}
\alpha = \frac{1}{\log \frac{q^2 + 1}{2q \cos \frac{\pi}{N}}}
\end{gather}
\begin{gather}
\label{eq:brickwork_beta}
\beta = 1 + \alpha \log \left[\frac{4 \cot ^2\frac{\pi}{n} \left(\left(q^2+1\right)^2-4 q^2 \cos
   \frac{2\pi}{n}\right)^2}{n^2 \left[q^8 - 2 \left(q^4-1\right) q^2 \cos \frac{2\pi}{n} - 1\right] + n\left[4 q^4 \cos\frac{4\pi}{n}-4
   q^4\right]}\right] 
\end{gather}
arising from the aforementioned calculation. By a similar strategy one may obtain asymptotic formulas for the periodic-boundary-condition brickwork and for the linear and circle graphs. We intend to give a more detailed derivation of this formula, including a tighter lower bound, a generalization to $t > 2$, and a similar strategy for obtaining upper bounds, in a future work.

This bound is included in the figure for reference. Although it is formally only a lower bound on the multiplicative error, it fits the data very well. The small deviations visible will shrink as $\epsilon \rightarrow 0$ and $n \rightarrow \infty$.  

\begin{figure}[H]
    \centering
    \begin{tikzpicture}
        \begin{scope}
            \node[anchor=north west,inner sep=0] (image_b) at (0,0)
            {\includegraphics[width=0.6\columnwidth]{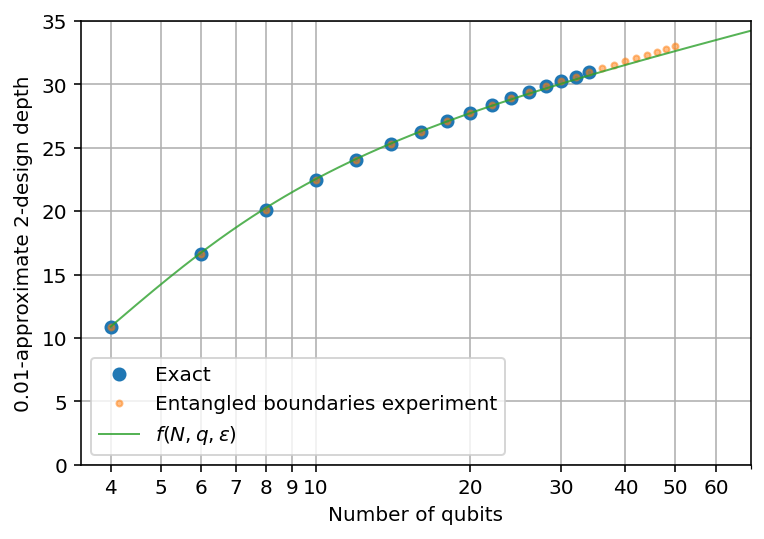}};
        \end{scope}
    \end{tikzpicture}
    \vspace*{-0.4cm}
    \caption{Depth needed for the brickwork to reach an $0.01$-approximate $2$-design, compared against two semi-empirical models. The entangled boundaries bound is tight everywhere tested, while the analytical form $f$ is merely a very good approximation.}
    \label{fig:brickwork_nscaling}
\end{figure}

\section{Fast architectures}
Here we present results on a few architectures which scramble especially quickly. 

\subsection{Prior work}
There has been a variety of prior work on shallow architectures which form good approximate $t$-designs. However, much of this work has focused on carefully-constructed circuits, with the goal of producing provable designs efficiently on a very large quantum computer \cite{metger_simple_2024, suzuki_more_2025}. Most recently, it was shown one can obtain an approximate $2$-design as fast as depth $O(\log \log n)$\cite{cui_unitary_2025}. However, this requires incorporating ancilla qubits, many non-Haar-random gates, and a large constant-factor overhead. If one further allows many-qubit gates or midcircuit measurement, this can be brought down to $O(1)$ \cite{foxman_random_2025}.

These results may be useful if one has a wishes to construct a unitary design on a quantum computer. However, here we are interested in studying the behavior of ``natural'' random circuits, with very little structure other than the geometric pattern of the gates. We thus restrict ourselves to local Haar-random gates and no ancillae. In this setting, Theorem 4 of ref. \cite{dalzell_random_2022} implies a lower bound on the number of gates per site needed,
\begin{align}
    \frac{s}{n} &\geq \frac{\log n - \log \frac{(q+1) \log \left(1 + 2 \epsilon\right)}{\log(q + 1)} }{\log (q^2 + 1)} 
    \\ &\sim \log_5 \frac{n}{\epsilon} - 0.801
\end{align}
(see Appendix~\ref{app:dalzell_bounds} for details). Ref. \cite{dalzell_random_2022} also asks which architecture gives the fastest possible anticoncentration, suggesting the parallel complete-graph defined below as a possible answer. The question we study here is quite similar. On the other hand, the fastest known provable examples are due to refs \cite{schuster_random_2025} and \cite{laracuente_approximate_2024}, both $O(\log n)$ (with large constants). These architectures are designed to be easy to prove theorems about, but it seems unlikely that they are especially fast scramblers in practice.

\subsection{Architectures}
\label{sec:fast_definitions}
Let us define a few more complicated distributions over circuit architectures. These are neither a single fixed arrangement of gates nor with a graph with gate locations sampled i.i.d.

Suppose we draw a random two-sided matching of the sites, then apply a layer of Haar-random 2-site gates to those pairs in parallel. In other words, we sample a random complete layer, i.e. a random set of $\frac{N}{2}$ gates such that each site is acted on by exactly one gate. This is the \textbf{parallel complete-graph} (PCG) architecture. This is the architecture suggested by ref. \cite{dalzell_random_2022} as a possible ``fastest anticoncentrator.''

Suppose we draw layers as in the parallel complete-graph architecture, except that we require each adjacent pair of layers to form a connected block. This guarantees, for example, that we never ``waste'' a gate by repeating a gate from the previous layer. Avoiding this kind of waste increases the speed of scrambling. This architecture is similar to applying a single period of 1D brickwork to a random permutation of the sites, so I'll call it the \textbf{permuted brickwork} (PB). Like the brickwork, this architecture can be shown to have a PSD vectorization if the depth is odd. It is unclear if the vectorization is PSD at even depths. 

Suppose we instead keep the even-numbered layers the same every time, but draw the odd layers so that each adjacent pair forms a connected block. This architecture scrambles slower than the permuted brickwork, but it's a bit more numerically tractable (see Appendix \ref{app:algo_tricks}), so we can study it out to larger system sizes. We'll call it the \textbf{permuted brickwork with fixed evens} (PBFE).

\subsection{Results}
Figure \ref{fig:fast} shows $0.01$-approximate $2$-design depths for the architectures discussed above. 

\begin{figure}[H]
    \centering
    \begin{tikzpicture}
        \begin{scope}
            \node[anchor=north west,inner sep=0] (image_a) at (0,0)
            {\includegraphics[width=0.6\columnwidth]{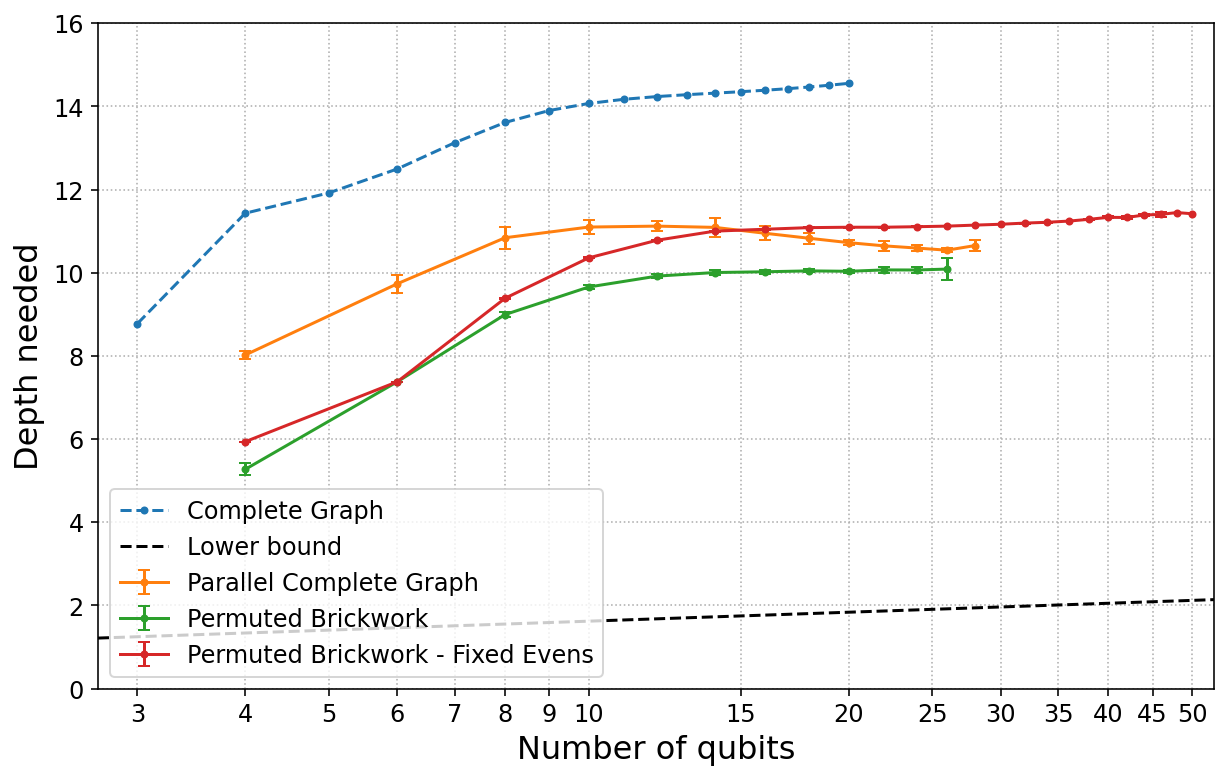}};
        \end{scope}
    \end{tikzpicture}
    \vspace*{-0.4cm}
    \caption{Approximate $2$-design depths for each of the fast architectures. All are faster than the complete graph, with Permuted Brickwork the fastest. Furthermore, depth appears to be quite flat out to $n = 50$ in at least the PBFE case.}
    \label{fig:fast}
\end{figure}

These architectures appear to form approximate $2$-designs much faster than any graph-sampled architecture tested. However, even the permuted brickwork is probably not the fastest possible architecture composed of Haar-random gates. It seems likely one could do even better with longer lookback periods (e.g. refusing to repeat not just the previous layer, but any of the 5 previous layers). One interesting question is if there is any optimal ensemble. For example, a Boolean hypercube architecture might be another interesting candidate to consider. Another is whether these are faster in practice than the constructions of refs. \cite{cui_unitary_2025, schuster_random_2025, laracuente_approximate_2024}. It is unclear if the large constant factors in those cases are artifacts of the proofs or truly essential. 

\section{Anticoncentration vs. 2-designs}
\label{sec:anticoncentration}
Ref.~\cite{heinrich_anti-concentration_2026} proves that anticoncentration and being a \textit{state} 2-design are essentially the same. Can this result be extended to the unitary case? For a unitary ensemble, anticoncentration asks about indistinguishability from the Haar measure by looking at a particular observable (the collision probability). An approximate $2$-design, on the other hand, requires that every choice of observable be hard to distinguish from the Haar measure. This raises a basic question: Do all observables converge at essentially the same rate? Is the behavior of the collision probability generic, or are there other classes of observables which are much slower to scramble? One may also view this question in the language of Eq \ref{eq:experimental_multerr_hermitian}, where collision probability corresponds to the choice $\vec{a} = \vec{0}$. In this language, we ask: Does the whole diagonal of the moment operator converge to its Haar value in roughly the same way, or are some of the elements special?

\subsection{Which experiments are optimal?}
\begin{figure}[H]
    \centering
    \begin{tikzpicture}
        \begin{scope}
            \node[anchor=north west,inner sep=0] (image_a) at (0,0)
            {\includegraphics[width=0.4\columnwidth]{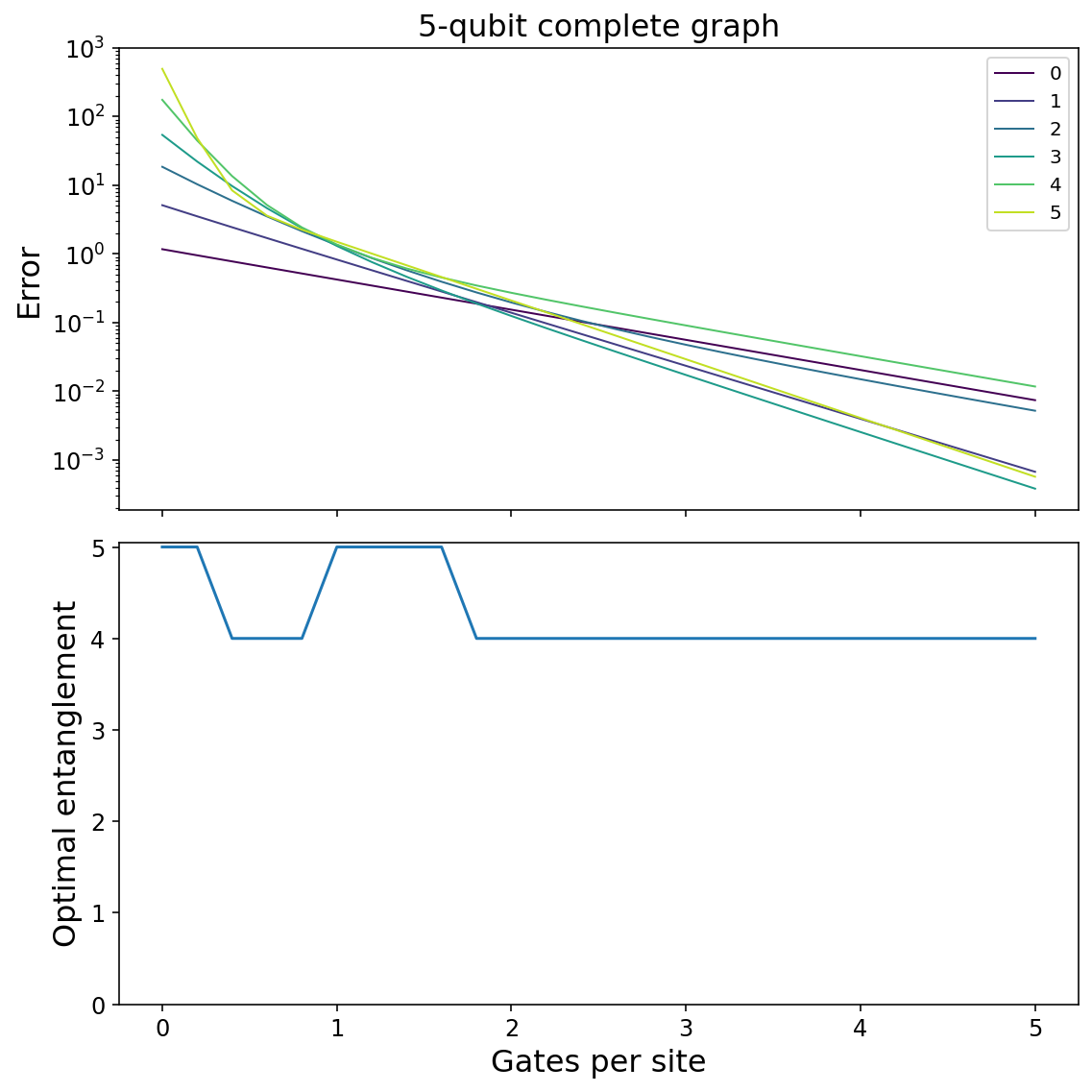}};
        \end{scope}
        \node [anchor=north west] (note) at (-0.2,0) {\small{\textbf{a)}}};
        \begin{scope}
            \node[anchor=north west,inner sep=0] (image_b) at (7.4,0)
            {\includegraphics[width=0.4\columnwidth]{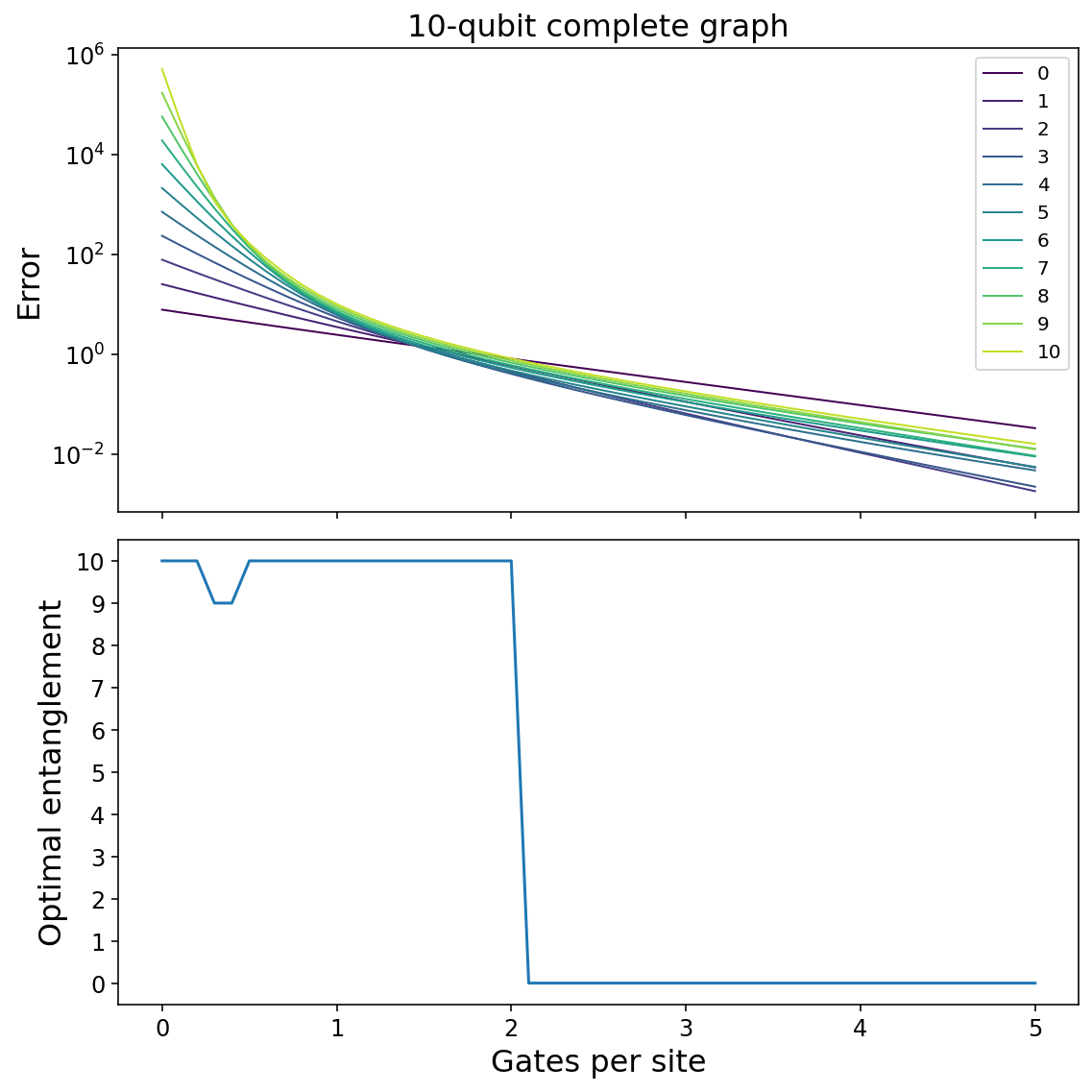}};
        \end{scope}
        \node [anchor=north west] (note) at (7.2,0) {\small{\textbf{b)}}};
    \end{tikzpicture}
    \caption{Error ratios and optimal experiments for complete-graph architectures. Upper panels show the error against the Haar measure for each choice of $\vec{a}$ and various circuit sizes. The legend indicates Hamming weights of $\vec{a}$, or equivalently the between-copy entanglement entropy of $\rho_{\vec{a}}$. Lower panels show the Hamming weight of the optimal $\vec{a}$ at each circuit size.
    (a) On 5 qubits, the optimal experiment always involves preparing singlet states on 4 or 5 of the sites. (b) With 10 qubits, singlet states are optimal below circuit size 20. Above circuit size 20, the collision probability (i.e. preparing product states on all sites) is better. However, the all-singlets experiment remains the second-best option.}
    \label{fig:complete_anticoncentration}
\end{figure}

Figure \ref{fig:complete_anticoncentration} shows the trajectories of all possible experiments (i.e. all $\vec{a} \in \{0,1\}^n$) for the complete graph. Since this ensemble is invariant under any permutation of the sites, we can label experiments only by the total number of singlet states prepared. This is the same as the Hamming weight of $\vec{a}$ or the entanglement entropy $S_E$ between the two copies (in bits). The error associated with experiment $\vec{a}$ is 
\begin{gather}
    \frac{\tr \left[\rho_{\vec{a}} \Phi_\varepsilon\left(\rho_{\vec{a}}\right)\right]}{\tr \left[\rho_{\vec{a}} \Phi_\text{Haar}\left(\rho_{\vec{a}}\right)\right]} - 1
\end{gather}

The main lesson of this figure is that the situation is quite complicated. Even for the complete graph, which is very symmetric, there appears to be no general pattern 
governing the optimal experiment. There are three loose trends which seem to hold widely:
\begin{itemize}
    \item At early times, highly-entangled experiments dominate. In particular, consider depth 0. The corresponding circuit is a tensor product of $n$ single-site unitaries. In this case one can show that the optimal experiment is $\vec{a} = \vec{1}$, which for qubits gives a multiplicative error $\sim 3^n$ times larger than the collisional error. 
    \item At late times, experiments with even Hamming weight all decay at the same rate (presumably corresponding to the spectral gap of the moment operator). Experiments with odd Hamming weight decay faster. In other words, when $\sum_i a_i$ is odd, then $\vectorize\left( \rho_{\vec{a}}\right)$ is orthogonal to the dominant eigenspace of the $\vectorize \Phi_\mathcal{E}$. 
    \item When there are 7 or more qubits and 2 or more gates per site, the collision probability dominates. 
\end{itemize}
For small systems and shallow circuits, however, we see only anarchy. 

\subsection{Does making it bigger help?}
The complete graph is at least well-behaved when the circuit is deep and wide enough. Is there a general rule that the collision probability determines the approximate $2$-design depth in some suitable limit? 

Figure \ref{fig:star_anticoncentration} compares anticoncentration and approximate $2$-design depths for the star graph. More precisely, it shows the number of gates required for both the multiplicative error
\begin{gather}
    \max_{\vec{a} \in \{0,1\}^{n}} \frac{\tr \left[\rho_{\vec{a}} \Phi_\varepsilon\left(\rho_{\vec{a}}\right)\right]}{\tr \left[\rho_{\vec{a}} \Phi_\text{Haar}\left(\rho_{\vec{a}}\right)\right]} - 1
\end{gather}
and the \textbf{collisional error}
\begin{gather}
    \frac{\tr \left[\rho_{\vec{0}} \Phi_\varepsilon\left(\rho_{\vec{0}}\right)\right]}{\tr \left[\rho_{\vec{0}} \Phi_\text{Haar}\left(\rho_{\vec{0}}\right)\right]} - 1
\end{gather}
to reach $0.01$.

\begin{figure}[H]
    \centering
    \begin{tikzpicture}
        \begin{scope}
            \node[anchor=north west,inner sep=0] (image_a) at (0,0)
            {\includegraphics[width=0.4\columnwidth]{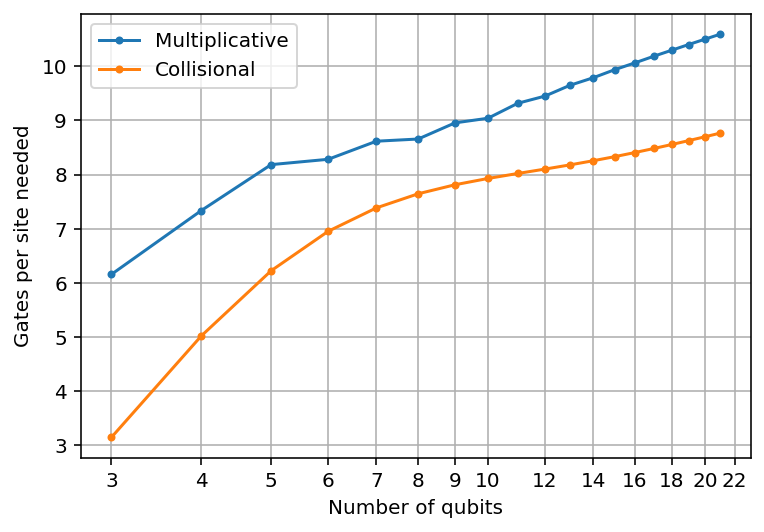}};
        \end{scope}
        \node [anchor=north west] (note) at (0,0) {\small{\textbf{a)}}};
        \begin{scope}
            \node[anchor=north west,inner sep=0] (image_b) at (7,0)
            {\includegraphics[width=0.4\columnwidth]{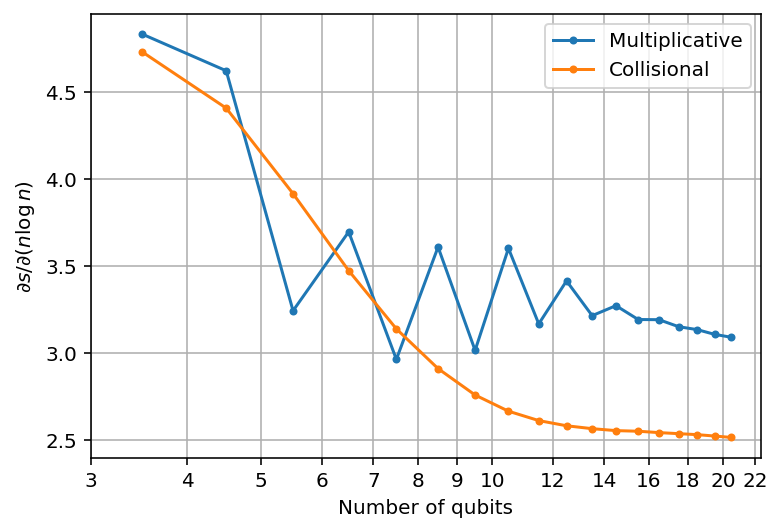}};
        \end{scope}
        \node [anchor=north west] (note) at (7.2,0) {\small{\textbf{b)}}};
    \end{tikzpicture}
    \caption{Anticoncentration vs. 2-design-ness for the star graph. (a) Depths needed for the multiplicative and collisional errors to reach $0.01$ for the star graph at various system sizes. (b) Slopes, estimated from finite differences. The two slopes appear to converge towards different constant levels. }
    \label{fig:star_anticoncentration}
\end{figure}

The gap between anticoncentration depth and approximate $2$-design depth appears to get larger as $n$ increases. This suggests that even with large $n$ or small $\epsilon$, it is not true that the anticoncentration depth and the approximate $2$-design depth are necessarily close together.

For the star, the optimal experiment generally involves preparing the entangled (antisymmetric) state on all of the points of the star. The parity of the center qubit depends on $n$; it should be chosen so that $\vec{a}$ is odd. It seems generally that optimal experiments involve entangled states near edges of the geometry and product states in the bulk.

We've seen that the scaling approximate $2$-design depth is controlled mostly by the norm of the projection of $\rho_{\vec{a}}$ into the dominant eigenspace. The collision probability can converge much faster than other observables if and only if $\vectorize \left(\rho_{\vec{0}}\right)$ is nearly orthogonal to the dominant eigenspace. It may be possible to understand these results more clearly by determining the dominant eigenvectors of the star graph.

\subsection{Anticoncentration of brickworks}
We see that the general situation is complicated. However, we can at least say something very concrete for 1D brickwork circuits. Figure \ref{fig:brickwork_anticoncentration} shows the (interpolated) circuit depth required for both the multiplicative and collisional error to reach 0.01, for 1D brickwork architectures with both open and periodic boundary conditions.

\begin{figure}[h]
    \centering
    \begin{tikzpicture}
        \begin{scope}
            \node[anchor=north west,inner sep=0] (image_a) at (0,0)
            {\includegraphics[width=0.45\columnwidth]{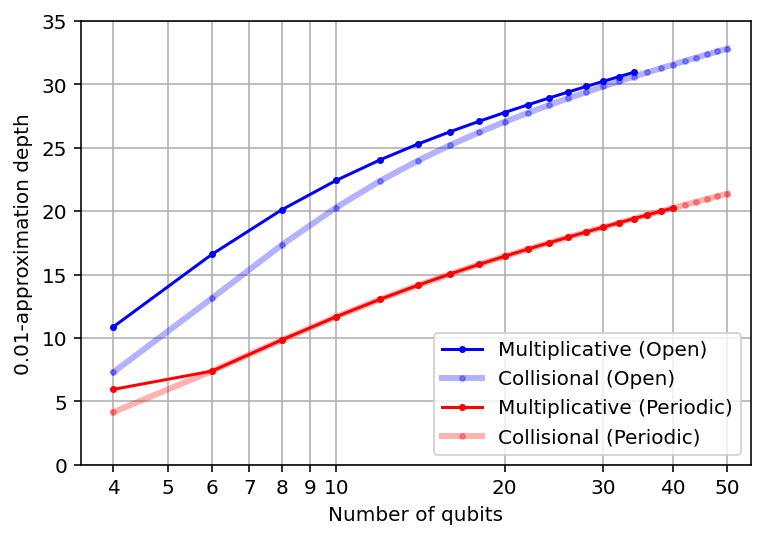}};
        \end{scope}
        \node [anchor=north west] (note) at (-0.1,0) {\small{\textbf{a)}}};
        \begin{scope}
            \node[anchor=north west,inner sep=0] (image_b) at (8,0){\includegraphics[width=0.45\columnwidth]{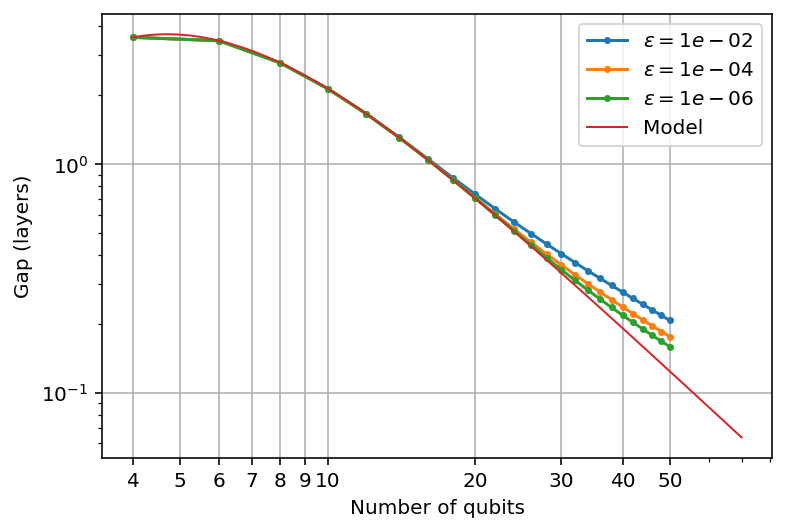}};
        \end{scope}
        \node [anchor=north west] (note) at (8.3,0) {\small{\textbf{b)}}};
    \end{tikzpicture}
    \caption{(a) $0.01$-approximate $2$-design depth and $0.01$-anticoncentration depths for 1D brickworks, with open and periodic boundary conditions. We see that anticoncentration and being a $2$-design are usually equivalent with periodic boundaries, but inequivalent with open boundaries. (b) The gap between the depths needed to reach relative error $\epsilon$ for the entangled boundaries experiment and the collision probability, for the open-boundary case. The gap converges to the prediction of the single-eigenvector model as $\epsilon \rightarrow 0$.}
\label{fig:brickwork_anticoncentration}
\end{figure}

 With open boundary conditions, the multiplicative error is dominated by the ``entangled boundaries'' state, 
\[\vec{a} = \begin{bmatrix} 
1 & 0 & 0 & \dots & 0 & 0 & 1 
\end{bmatrix}\]
With periodic boundary conditions, the collision probability itself dominates for all $n \geq 6$, and so the two curves coincide exactly. 

To understand the gap seen in the open case, we can work out an analogue of Equation \ref{eq:brickwork_formula_template} for the collision probability. With open boundary conditions, it turns out that the contribution of the dominant eigenspace is 
\begin{align}
    f(N,q,\epsilon) &= \alpha \left(\log n - \log \epsilon\right) + \beta
\end{align}
with parameters
\begin{gather}
\alpha = \frac{1}{\log \frac{q^2 + 1}{2q \cos \frac{\pi}{n}}}
\end{gather}
\begin{gather}
\beta = 1 + \alpha \log \left[\frac{4 \cot ^2\frac{\pi}{n} \left(q^2-1\right)^4}{n^2 \left[q^8 - 2 \left(q^4-1\right) q^2 \cos \frac{2\pi}{n} - 1\right] + n\left[4 q^4 \cos\frac{4\pi}{n}-4
   q^4\right]}\right] 
\end{gather}
which differ from Equations \ref{eq:brickwork_alpha} and \ref{eq:brickwork_beta} only slightly, in the numerator of $\beta$. If we take $q = 2$ and expand for large $n$, we find eventually that the difference in depths is
\begin{gather}
    \Delta = \frac{64 \pi^2}{9\log \frac{5}{4}} \frac{1}{n^2} + O\left(\frac{1}{n^3}\right) \approx \frac{314.52}{n^2}
\end{gather}
The analogous calculation for periodic boundary conditions is just $\Delta = 0$.

\section{Conclusion}
Previous work on approximate unitary designs has generally focused on proving asymptotic bounds. While this has resulted in much progress, it is rarely clear how well a provable bound corresponds to the actual behavior of the ensemble. Here we instead give a variety of exact calculations for finite system sizes. Although our results can't formally say anything about large $n$, they are in many cases strongly suggestive. Data like this helps illustrate the relationship between what we can prove and what is true.

We do include a few pure theoretical contributions. First, we show that the optimal experiment which distinguishes a given random circuit architecture from the Haar measure is highly constrained. Second, we give a relatively tractable algorithm for determining the approximate $2$-design depths of suitable circuits. And third, we prove that at least some graph families require $O(n^2)$ gates to form an approximate $t$-design. 

There are several directions in which one may extend our algorithm. First and most obvious is an extension to $t > 2$. In that case the irreps are no longer one-dimensional, and so the Choi matrix is merely block-diagonal, which presents some difficulties. A second question is whether they can be extended to groups other than the unitary group. Recent results have shown that the formation of designs over more constrained groups may behave quite differently \cite{grevink_will_2025, west_no-go_2025, liu_unitary_2024, mitsuhashi_unitary_2025, li_designs_2024, haah_short_2025}. It seems likely that Theorem \ref{thm:multerr_expdef} can be extended to other groups with suitable structure. A third question is whether similar formulas exist for additive or measurable error. 

The bulk of this paper concerns our numerical results. It appears that all graphs require at least $\Omega(n \log n)$ gates and at most $O(\log n)$ connections to form an approximate $2$-design. We furthermore suggest that the linear and complete graphs are, as has often been guessed \cite{dalzell_random_2022}, most likely extremal. This greatly constrains the influence of graph geometry on the scrambling rate. At least three important open questions in this direction remain:
\begin{itemize}
    \item Can our conjectures be proven?
    \item Can our complicated measure of connectedness be replaced by some simpler property of the graph, e.g. the ratio of the minimum cut to the total edge count?
    \item Can these observations about graphs be extended to arbitrary arrangements of gates, similar to the conjectures discussed in ref.~\cite{belkin_approximate_2024}?
\end{itemize}

For brickworks, we give an equation for the approximate $2$-design depth which appears to be quite accurate in practice. Here the most important remaining question is how the brickwork behaves at $t > 2$. In addition, of course, one would like to prove the correctness of our semi-empirical formula. 

The fast architectures we study seem to scramble much faster than other known ensembles. One interesting question is whether there exists any nicely-structured fastest ensemble, either with or without the restriction to Haar-random local gates \cite{suzuki_more_2025}. In practice it would of course be useful to know the quickest route to an approximate design on a modest-sized quantum device. 

Finally, we show that recent results on state designs from anticoncentration are likely to be difficult to extend to the unitary case. The ratio between collisional and multiplicative errors can grow arbitrarily large. On the other hand, Theorem \ref{thm:multerr_expdef} gives a conceptual connection between anticoncentration and approximate $2$-design-ness. Perhaps this result will offer an alternative route to establishing a log-depth bound for suitably structured circuits.

\paragraph{Acknowledgements} 
D.B. acknowledges discussions with Nicholas LaRacuente and Aram Harrow which persuaded him that numerical results on $2$\textsuperscript{nd} moments were worth pursuing. He also acknowledges extensive discussions with Markus Heinrich and Jonas Haferkamp which led to clarifications of Theorem \ref{thm:multerr_expdef} and motivated Section \ref{sec:anticoncentration}, and thanks them for sharing earlier drafts of ref.~\cite{heinrich_anti-concentration_2026}. J.A. is supported by a postdoctoral fellowship as part of the niversit\'e de Montr\'eal and the Institut Courtois.  This material is based upon work supported by the U.S. Department of Energy, Office of Science, National
Quantum Information Science Research Centers.

\printbibliography

\begin{appendices}
\section{Multiplicative errors at $t = 2$} \label{app:multiplicative_errors}
\subsection{Basic setup}
Suppose we have $n$ sites, each with a local Hilbert space of dimension $q$. We have some distribution $\varepsilon$ over the unitary group of which we are studying the $t$th moment. Later we will specialize to $t = 2$.

The $t$th moment operator of some distribution $\varepsilon$ over the unitary group is a quantum channel given by 
\begin{gather}
        \Phi_\varepsilon^{(t)}(\rho) = \mathbb{E}_{U \sim \varepsilon}\left[(U^\dagger)^{\otimes t} \! \rho (U)^{\otimes t}\right]
\end{gather}
The multiplicative distance between two channels $A,B$ is defined to be the smallest \(\epsilon\) such that $(1 + \epsilon) B - A$ and $A - (1 - \epsilon) B$ are both completely positive maps. 

We'll also use the Choi isomorphism, 
\begin{gather}
    \choi(\mathcal{N}) = \left[\mathcal{N} \otimes \mathcal{I}\right]\left(\frac{1}{d} \sum_{i=1}^d \ket{i} \otimes \ket{i} \sum_{i=1}^d \bra{i} \otimes \bra{i}\right)
\end{gather}
Since complete positivity of a channel is equivalent to positive semidefiniteness of the corresponding Choi state, we can rephrase this as the smallest $\epsilon$ such that 
\begin{gather}
    (1 + \epsilon) \choi(B) \succeq \choi(A) \succeq (1-\epsilon)\choi(B)
\end{gather}
We will show that this expression becomes especially simple for second moment operators.

\subsection{From Choi Positivity to Likelihood Ratios}
Define functions $a(\mathbf{v}) = \mathbf{v}^\dagger \choi(A) \mathbf{v}$ and likewise for $B$. The condition
\begin{gather}
    (1 + \epsilon) \choi(B) \succeq \choi(A) \succeq (1-\epsilon)\choi(B)
\end{gather}
is then equivalent to 
\begin{gather}
    (1 + \epsilon) b(\mathbf{v}) \geq a(\mathbf{v}) \geq (1-\epsilon)b(\mathbf{v})
\end{gather}
which implies
\begin{gather}
    \epsilon  = \max_{\mathbf{v}} \left|\frac{a(\mathbf{v})}{b(\mathbf{v})} - 1\right|
\end{gather}
We now show 
\begin{gather}
    \label{eq:choi_ratio_target}
    \max_{\mathbf{v}} \left|\frac{a(\mathbf{v})}{b(\mathbf{v})} - 1\right| = \max_{\rho, \Pi} \left |\frac{\tr \left(\Pi \left[A \otimes I\right](\rho)\right)}{\tr \left(\Pi \left[B \otimes I\right](\rho) \right)} - 1\right |
\end{gather}
The Choi operator acts on two copies of the Hilbert space. We decompose $\mathbf{v} = \sum_{i,j} v_{ij} \ket{i} \otimes \ket{j}$. If we then choose $\rho_{ijkl} \propto v_{ij} v_{kl}$ and $\Pi_{ijkl} = \delta_{ij} \delta_{kl}$, we have
\begin{gather}
    a(\mathbf{v}) = \tr \left(\Pi [A \otimes I](\rho)\right)
\end{gather}
which establishes that the left-hand side of Eq.~\ref{eq:choi_ratio_target} is no larger than the right-hand side. On the other hand, given an arbitrary $\rho$ and $\Pi$, we may by convexity find rank-1 $\rho' = \bra{\psi}\ket{\psi}$ and $\Pi' = \bra{\phi}\ket{\phi}$ such that 
\begin{gather}
    \frac{\tr \left(\Pi \left[A \otimes I\right](\rho)\right)}{\tr \left(\Pi \left[B \otimes I\right](\rho) \right)} \leq \frac{\tr \left(\Pi' \left[A \otimes I\right](\rho')\right)}{\tr \left(\Pi' \left[B \otimes I\right](\rho') \right)}
\end{gather}
and similarly may find other $\rho''$, $\Pi''$ which give a lower bound. We may then take $\mathbf{v} = \ket{\psi} \otimes \ket{\phi}$, which proves that the right-hand side of Eq.~\ref{eq:choi_ratio_target} is no larger than the left. It follows that they must be equal.

\subsection{Choi isomorphism on permutation basis}

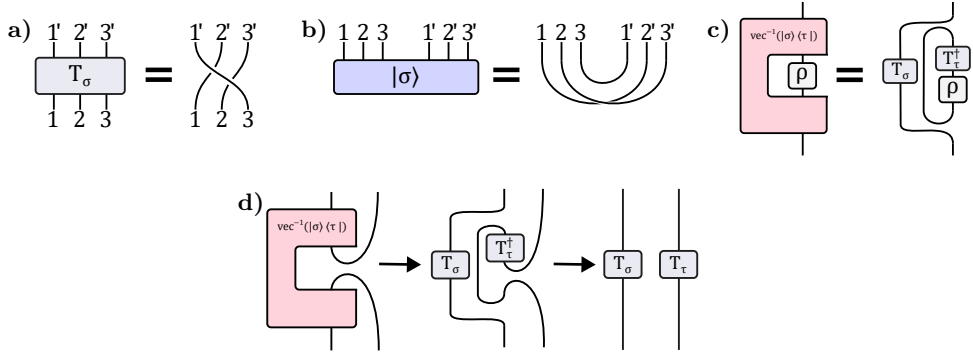
\begin{figure}[h]
    \centering
    \begin{tikzpicture}
        \begin{scope}
            \node[anchor=north west,inner sep=0] (image_a) at (0,0)
            {\includegraphics[width=0.75\columnwidth]{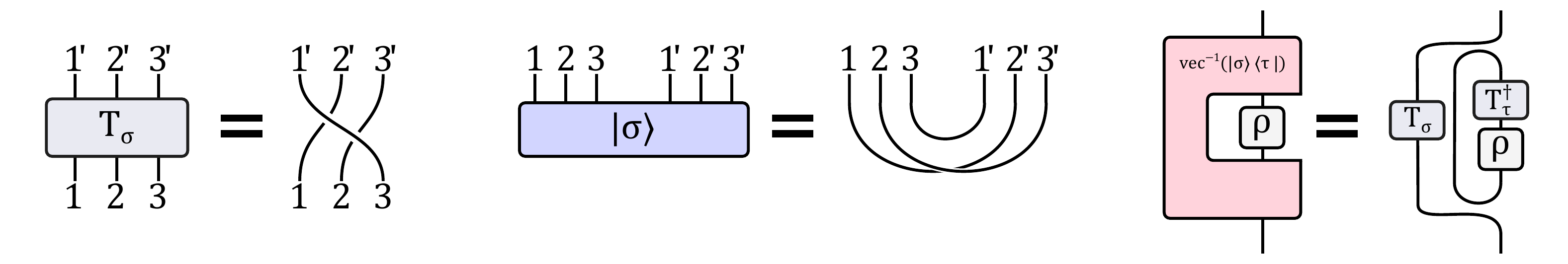}};
            \node [anchor=north west] (note) at (-0.1,-0.2) {\small{\textbf{a)}}};
            \node [anchor=north west] (note) at (3.8,-0.2) {\small{\textbf{b)}}};
            \node [anchor=north west] (note) at (9.2,-0.2) {\small{\textbf{c)}}};
        \end{scope}
    \end{tikzpicture}
    \vspace*{-0.4cm}
    \begin{tikzpicture}
        \begin{scope}
            \node[anchor=north west,inner sep=0] (image_a) at (0,0)
            {\includegraphics[width=0.4\columnwidth]{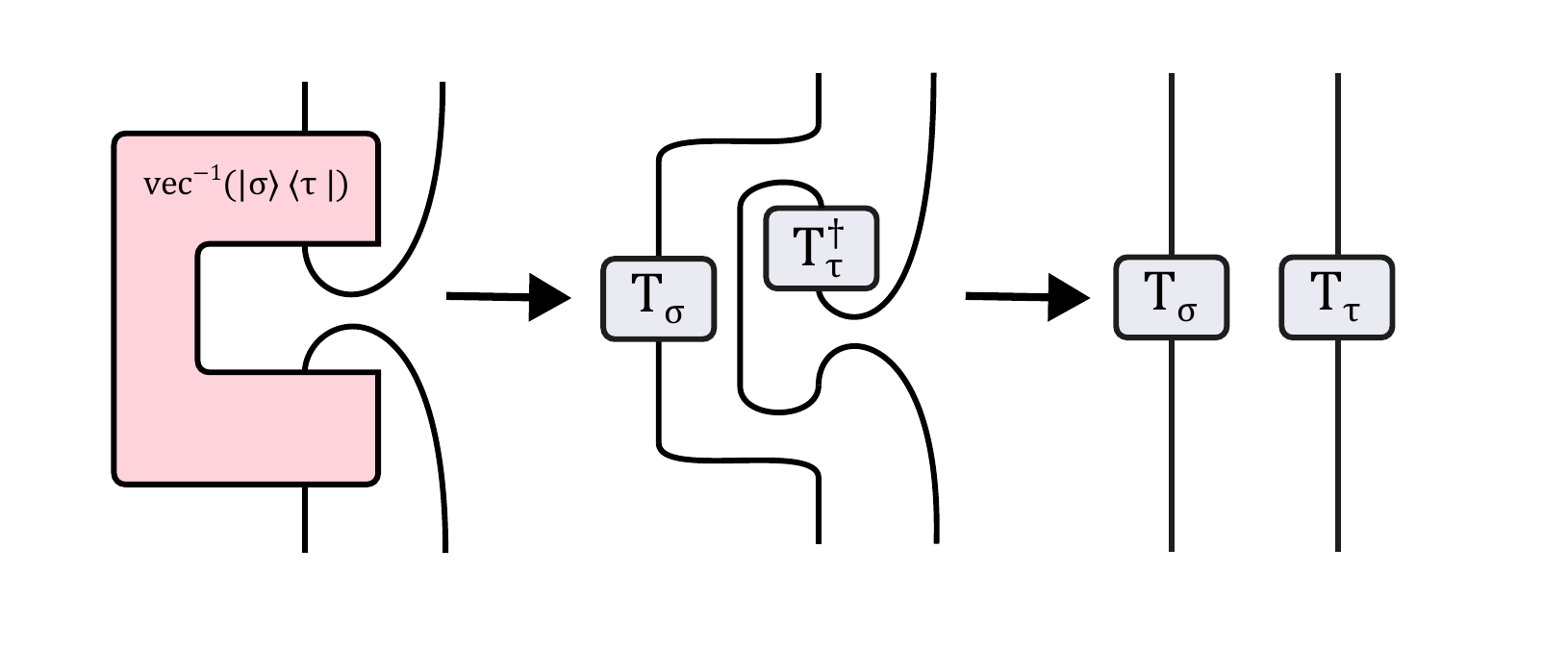}};
            \node [anchor=north west] (note) at (0,-0.25) {\small{\textbf{d)}}};
        \end{scope}
    \end{tikzpicture}
    \vspace*{-0.2cm}
    \caption{Tensor network depiction of a permutation operator $T_\sigma$ (a) and the corresponding permutation state $\ket{\sigma}$ (b) on a three-copy Hilbert space, with $\sigma = (123)$. (c): Tensor network depiction of the channel $\vectorize^{-1}(\ket{\sigma}\bra{\tau})$ acting on an arbitrary density matrix $\rho$, as per Eq~\ref{eq:inverse_vectorized_channel}. (d): Choi isomorphism of the channel $\vectorize^{-1}(\ket{\sigma}\bra{\tau})$, which can be decomposed into $T_\sigma \otimes T_\tau$.}
    \label{fig:permutation_states_choi}
\end{figure}

Circuits composed of Haar-random gates induce a measure $\varepsilon$ which is invariant under single-site unitaries. It follows that the image of the moment operator lies in the single-site commutant of $\mathcal{U}(q)^{\otimes t}$. This subspace is spanned by the permutation operators
\begin{gather}
T_\sigma = \sum_{\vec{i} \in \mathbb{Z}_q^t} \ket{\sigma(\vec{i})} \bra{\vec{i}}
\end{gather}
We will work with the vectorization map $\vectorize$, under which (Fig.~\ref{fig:permutation_states_choi}a-b) 
\begin{gather}
\ket{\sigma} \equiv \vectorize(T_\sigma)   = \frac{1}{\sqrt{q}^t} \sum_{\vec{i} \in \mathbb{Z}_q^t} \ket{\vec{i}} \otimes \ket{\sigma(\vec{i})}
\end{gather}
In this basis, the vectorization of the channel corresponds to some matrix, i.e. 
\begin{gather}
    \vectorize(\Phi) = \sum_{\sigma_1 \dots \sigma_N, \tau_1 \dots \tau_N} M_{\sigma_1 \dots \sigma_N, \tau_1 \dots \tau_N}\ket{\sigma_1 \dots \sigma_N} \bra{\tau_1 \dots \tau_N}
\end{gather}
for some coefficients $M$. By linearity, the corresponding Choi state may then be expressed as
\begin{gather}
    \choi(\Phi) = \sum_{\sigma_1 \dots \sigma_N, \tau_1 \dots \tau_N} M_{\sigma_1 \dots \sigma_N, \tau_1 \dots \tau_N} \choi \circ \vectorize^{-1}\left(\ket{\sigma_1 \dots \sigma_N} \bra{\tau_1 \dots \tau_N}\right)
\end{gather}

How does the Choi map act on these permutation basis states? On a single site, we have
\begin{align}
    \left[\vectorize^{-1}\left(\ket{\sigma} \bra{\tau}\right)\right](\rho) &= T_\sigma \tr(T_{\tau}^\dagger \rho) \label{eq:inverse_vectorized_channel}
\end{align}
and so (Fig.~\ref{fig:permutation_states_choi}d)
\begin{align}
    \choi \circ \vectorize^{-1}\left(\ket{\sigma} \bra{\tau}\right) &= (T_\sigma  \otimes I) \tr_A\left[ (T_\tau^\dagger \otimes I) \left(\frac{1}{q^t} \sum_{j=1}^{q^t} \ket{j} \otimes \ket{j} \sum_{k=1}^{q^t} \bra{k} \otimes \bra{k}\right)\right]
    \\
    &= \frac{1}{q^{2t}} \sum_{\vec{i},\vec{j},\vec{k} \in \mathbb{Z}_q^t} (\rho_\sigma \otimes I)\tr_A\left[\ket{\vec{i}}\bra{\tau(\vec{i})}\ket{\vec{j}\vec{j}}\bra{\vec{k}\vec{k}} \right]
    \\
    &= \frac{1}{q^{2t}} \sum_{\vec{i},\vec{j},\vec{k} \in \mathbb{Z}_q^t} (\rho_\sigma  \otimes I) \braket{\vec{k}|\vec{i}} \braket{\tau(\vec{i})|\vec{j}} \left(I \otimes \ket{\vec{j}}\bra{\vec{k}}\right)
    \\
    &= \frac{1}{q^{2t}} (\rho_\sigma  \otimes I) \left(I \otimes \sum_{\vec{i} \in \mathbb{Z}_q^t} \ket{\tau(\vec{i})}\bra{\vec{i}}\right)
    \\
    &= T_\sigma  \otimes T_\tau
\end{align}
More generally, on multiple sites, we see
\begin{gather}
    \choi\circ \vectorize^{-1}\left(\ket{\sigma_1 \dots \sigma_N} \bra{\tau_1 \dots \tau_N}\right) = \bigotimes_i (T_{\sigma_i} \otimes T_{\tau_i})
\end{gather}
so the Choi state may be represented as a linear combination of products of twist operators.

\subsection{Decomposition of Choi state into irreps}
These twist operators are a representation of the symmetric group $S_t$. It follows that the Choi state belongs to a representation of the algebra $S_t^{2N}$. The eigenspaces of any such algebra element may be decomposed into irreducible representations. In particular, if $V_\nu$ is an irreducible representation of $S_t$ labeled by a partition $\nu$, we may decompose the twist operator into irreps 
\begin{gather}
    T_\sigma \cong \bigoplus_{\nu \vdash t, |\nu| \leq q} I_{r_\nu} \otimes V_\nu
\end{gather}
for some multiplicities $r_\nu$. 

Let $\alpha_i,\beta_i,  i \in \{1...N\}$ label a set of $2N$ irreducible representations of $S_t$, with corresponding representations $V_{\alpha_i}(\sigma)$. Then the eigenvalues of $ \choi(\Phi_\varepsilon) - (1 \pm \epsilon)\choi(\Phi_\text{Haar})$ are precisely the eigenvalues of 
\begin{gather}
    \sum_{\vec{\sigma},\vec{\tau}} (M^\varepsilon_{\vec{\sigma},\vec{\tau}} - (1 \pm \epsilon)M^\text{Haar}_{\vec{\sigma},\vec{\tau}}) \bigotimes_i \left(V_{\alpha_i}(\sigma_i) \otimes V_{\beta_i}(\tau_i)\right)
\end{gather}
It follows that 
$(1+\epsilon)\choi(\Phi_\text{Haar}) \succeq \choi(\Phi_\varepsilon) \succeq (1-\epsilon)\choi(\Phi_\text{Haar})$ if and only if 
\begin{gather}
    \sum_{\vec{\sigma},\vec{\tau}} (M^\varepsilon_{\vec{\sigma},\vec{\tau}} - (1 - \epsilon)M^\text{Haar}_{\vec{\sigma},\vec{\tau}}) \bigotimes_i \left(V_{\alpha_i}(\sigma_i) \otimes V_{\beta_i}(\tau_i)\right) \succeq 0
\end{gather}
and
\begin{gather}
    \sum_{\vec{\sigma},\vec{\tau}} ((1 + \epsilon)M^\text{Haar}_{\vec{\sigma},\vec{\tau}} - M^\varepsilon_{\vec{\sigma},\vec{\tau}}) \bigotimes_i \left(V_{\alpha_i}(\sigma_i) \otimes V_{\beta_i}(\tau_i)\right) \succeq 0
\end{gather}
for every choice of $\vec{\alpha},\vec{\beta}$. 

\subsection{Specializing to $t = 2$}
In the case $t = 2$, there are only two permutations, which we label $I$ and $S$. There are also only two irreps, trivial and sign. Both are of dimension one. It follows that 
\begin{gather}
    \bigotimes_i \left(V_{\alpha_i}(\sigma_i) \otimes V_{\beta_i}(\tau_i)\right) = \prod_{i=1}^N a_i^{|\sigma_i|} b_{i}^{|\tau_i|}
\end{gather}
where $a_i, b_i \in \{-1,1\}$ now label the sign and trivial irreps on site $i$. Because the irreps are all 1D, we immediately obtain the eigenvalues of $\choi(\Phi_\varepsilon) - (1 \pm \epsilon)\choi(\Phi_\text{Haar})$ as
\begin{gather}
    \sum_{\vec{\sigma},\vec{\tau}} (M^\varepsilon_{\vec{\sigma},\vec{\tau}} - (1 \pm \epsilon)M^\text{Haar}_{\vec{\sigma},\vec{\tau}}) \prod_{i=1}^N a_i^{|\sigma_i|} b_{i}^{|\tau_i|}
\end{gather}
for any particular choice of $\vec{a},\vec{b}$. Hence, the positive-semidefinite condition becomes the scalar condition that 
\begin{gather} \label{eq:intermediate_psd_condition}
    \sum_{\vec{\sigma},\vec{\tau}} (M^\varepsilon_{\vec{\sigma},\vec{\tau}} - (1 - \epsilon)M^\text{Haar}_{\vec{\sigma},\vec{\tau}}) \prod_{i=1}^N a_i^{|\sigma_i|} b_{i}^{|\tau_i|} \geq 0 \geq \sum_{\vec{\sigma},\vec{\tau}} (M^\varepsilon_{\vec{\sigma},\vec{\tau}} - (1 + \epsilon)M^\text{Haar}_{\vec{\sigma},\vec{\tau}}) \prod_{i=1}^N a_i^{|\sigma_i|} b_{i}^{|\tau_i|}
\end{gather}
for all $\vec{a},\vec{b}$. With some algebra we may rearrange this condition to 
\begin{gather}
     \epsilon \geq \left|\frac{\sum_{\vec{\sigma},\vec{\tau}}M^\mathcal{E}_{\vec{\sigma},\vec{\tau}}\prod_{i=1}^N a_i^{|\sigma_i|} b_{i}^{|\tau_i|}}{\sum_{\vec{\sigma},\vec{\tau}}M^\text{Haar}_{\vec{\sigma},\vec{\tau}}\prod_{i=1}^N a_i^{|\sigma_i|} b_{i}^{|\tau_i|}} - 1\right|
\end{gather}
provided the corresponding eigenvalue of $\choi(\Phi_\text{Haar})$ was nonzero - however, if it was, the original positive semidefiniteness of $\choi(\Phi_\varepsilon)$ would automatically satisfy Equation~\ref{eq:intermediate_psd_condition} for all $\epsilon$. Therefore, we see that the multiplicative error is given by
\begin{gather} \label{eq:mult_error_as_eigenvalue_ratio}
    \epsilon = \max_{\vec{a},\vec{b} \in \{-1,1\}^{N}} \left|\frac{\sum_{\vec{\sigma},\vec{\tau}}M^\mathcal{E}_{\vec{\sigma},\vec{\tau}}\prod_{i=1}^N a_i^{|\sigma_i|} b_{i}^{|\tau_i|}}{\sum_{\vec{\sigma},\vec{\tau}}M^\text{Haar}_{\vec{\sigma},\vec{\tau}}\prod_{i=1}^N a_i^{|\sigma_i|} b_{i}^{|\tau_i|}} - 1\right|
\end{gather}

\subsection{Cobasis and tensor network view}
We now show that the multiplicative error corresponds to the largest matrix element of the moment operator in a particular basis. Following ref. \cite{allen_conditional_2025}, we define the \textbf{permutation cobasis} as
\begin{gather}
    \ket{\widetilde{\sigma}} = q^t\sum_{\tau} \text{Wg}(\sigma^{-1}\tau)\ket{\tau}
\end{gather}
so that (so long as $t \leq q$) we have
\begin{gather}
    \braket{\tau|\widetilde{\sigma}} = \delta_{\sigma \tau}
\end{gather}
Let us now define the family of (unnormalized) states
\begin{gather}
    \ket{\Psi(\vec{a})} = \bigotimes_{i=1}^N \ket{\widetilde{I}} + a_i \ket{\widetilde{S}}
\end{gather}
so that
\begin{gather}
\braket{\Psi(\vec{a})|\sigma_1...\sigma_N} = \prod_{i=1}^N a_i^{|\sigma_i|}
\end{gather}
and we may write
\begin{gather}
\sum_{\vec{\sigma},\vec{\tau}}M^\varepsilon_{\vec{\sigma},\vec{\tau}}\prod_{i=1}^N a_i^{|\sigma_i|} b_{i}^{|\tau_i|} = 
    \bra{\Psi(\vec{a})} \vectorize(\Phi_\varepsilon)\ket{\Psi(\vec{b})}
\end{gather}
This form is convenient both for numerical calculations and for theoretical analysis of scaling with depth, since a deeper circuit corresponds to a power of $\vectorize(\Phi_\varepsilon)$. For any given arrangement of gates, one may express $\vectorize(\Phi_\varepsilon)$ as a tensor network made up of local moment operators. The product state $\ket{\Psi(\vec{a})}$ then corresponds to a boundary condition for that network. 

\subsection{Eigenvalues of $\choi \left(\Phi_\text{Haar}\right)$}
For the Haar measure, 
\begin{gather}
    M^\text{Haar}_{\vec{\sigma},\vec{\tau}} = \delta_{\sigma_1...\sigma_N} \delta_{\tau_1...\tau_N} q^{Nt}\text{Wg}(\sigma_1 \tau_1^{-1}, q^N)
\end{gather}
so only $t!^2$ out of the $t!^{2N}$ entries are nonzero. We may compute
\begin{align}
    \bra{\Psi(\vec{a})} \vectorize(\Phi_\text{Haar})\ket{\Psi(\vec{b})} 
    &=  M^\text{Haar}_{I^{\otimes n},I^{\otimes n}} + 
    \left(\prod_i a_i\right)M^\text{Haar}_{S^{\otimes n},I^{\otimes n}} + 
     \left(\prod_i b_i\right)M^\text{Haar}_{I^{\otimes n},S^{\otimes n}} + 
     \left(\prod_i a_i b_i\right)M^\text{Haar}_{S^{\otimes n},S^{\otimes n}}
     \\ &= 
     \left(1 \!+\! (-1)^{\mathcal{P}\left(\vec{a}\right) + \mathcal{P}\left(\vec{b}\right)}\right)q^{Nt}\text{Wg}(I, q^N)\! + \!\left((-1)^{\mathcal{P}\left(\vec{a}\right)} \!+ \!(-1)^{\mathcal{P}\left(\vec{b}\right)}\right)q^{Nt}\text{Wg}(S, q^N)
\end{align}
where $\mathcal{P}(\vec{x})$ is the parity of $\sum_i x_i$. Therefore,
\begin{align} \label{eq:haar_definition}
    \bra{\Psi(\vec{a})} \vectorize(\Phi_\text{Haar})\ket{\Psi(\vec{b})} &= \begin{cases}
         \frac{2}{1 + q^{-N}} & \mathcal{P}\left(\vec{a}\right) = \mathcal{P}\left(\vec{b}\right) = 0\\
         \frac{2}{1 - q^{-N}} & \mathcal{P}\left(\vec{a}\right) = \mathcal{P}\left(\vec{b}\right) = 1 \\
         0 & \mathcal{P}\left(\vec{a}\right) \ne \mathcal{P}\left(\vec{b}\right)
     \end{cases}
\end{align}
Note that $\ket{\Psi}$ is not normalized. The corresponding physical outcome probabilities after normalization are $\frac{2}{q^n \pm 1}$. 

\subsection{$\vectorize\left(\Phi_\varepsilon - \Phi_\text{Haar}\right)$ is often PSD} \label{app:psd_circuits}
Suppose for now that $\vectorize \Phi_\varepsilon$ is positive-semidefinite. $\Phi_\text{Haar}$ is an orthogonal projector in to (a subspace of) the unit eigenspace of $\vectorize \Phi_\varepsilon$ , so the eigenvalues of $\vectorize\left(\Phi_\varepsilon - \Phi_\text{Haar}\right)$ are exactly those of $\vectorize \Phi_\varepsilon$, except that two of the $1$ eigenvalues have been replaced by $0$ \cite{allen_conditional_2025}. Clearly this remains positive semidefinite. 

Now, the assumption of positive-semidefiniteness doesn't hold for arbitrary circuits. We will show, however, that it holds for each class of circuits studied here. 

\paragraph{Graphs} The single-gate moment operator is an orthogonal projector into the single-site commutant, so it is PSD. A tensor product of a PSD operator with the identity is also PSD. The full moment operator for a graph-sampled architecture is a convex combination of such operators, so it is also PSD. 

\paragraph{Brickwork (odd depths)}
The vectorized moment operator of the brickwork architecture may be written $(L_O L_E)^d$, where $L_O$ and $L_E$ are the odd and even layers, respectively. Each layer is an orthogonal projector, so e.g. $L_O^2 = L_O^\dagger = L_O$. Suppose $d = 2k + 1$. Then we may write 
\begin{align}
    (L_O L_E)^d &= (L_O L_E)^k (L_O L_E)^k L_O 
    \\ &= (L_O L_E)^k L_O^2 (L_E L_O)^k
    \\ &= \left[(L_O L_E)^k L_O\right] \left[(L_O L_E)^k L_O\right]^\dagger
\end{align}
which is of the form $X^\dagger X$ and so positive semi-definite.

\paragraph{Fast architectures}
For the PCG architecture, the argument is essentially the same as for graphs. A single layer gives a moment operator which is a convex combination of tensor products of projectors, and a deeper circuit is just a power of the single-layer case. 

The PBFE architecture is more similar to the case of brickwork. We may duplicate each even layer to obtain a composition of 3-layer circuits. Each 3-layer circuit is a convex combination of 3-layer brickworks, so it is PSD. Consecutive 3-layer circuits are sampled independently, so their composition is also PSD. 

The trickiest case is the Permuted Brickwork. In this case consecutive layers are not independent. We use instead the following argument: Consider a periodic brickwork architecture with an odd number of layers. Condition on a particular choice for the middle layer. After conditiong on the layout of the middle layer, the first and second halves of the circuit are independent of each other, and their distributions are related by inversion. It follows that we may write the moment operator for a $k$-layer permuted brickwork as
\begin{gather}
    \vectorize \Phi_{\text{PB } k} = \mathbb{E}_{\text{middle}} \left[\vectorize \left(\Phi_{\text{PB } \frac{k - 1}{2}}\right)^{\dagger} \vectorize \left(\Phi_{\text{middle}} \right)\vectorize \left(\Phi_{\text{PB } \frac{k - 1}{2}} \right)\right]
\end{gather}
This is again a convex combination of PSD matrices.

\subsection{Diagonal dominance}
We now show that when $\vectorize\left(\Phi_\varepsilon - \Phi_\text{Haar}\right)$ is PSD, a case $\vec{a} = \vec{b}$ dominates the multiplicative error. 

The largest element of a PSD matrix occurs on the diagonal, and all diagonal elements are always positive. This is not quite enough to establish the desired result, since Equation \ref{eq:mult_error_as_eigenvalue_ratio} involves a ratio of elements. We must first split the error into
\begin{gather}
    \epsilon_\text{even} = \frac{1 + q^{-N}} {2}\max_{\vec{a},\vec{b} \text{ even}} \left|\bra{\Psi(\vec{a})} \vectorize\left(\Phi_\varepsilon - \Phi_\text{Haar}\right) \ket{\Psi(\vec{b})}\right|
\end{gather}
and
\begin{gather} 
    \epsilon_\text{odd} = \frac{1 - q^{-N}}{2} \max_{\vec{a},\vec{b} \text{ odd}}  \left|\bra{\Psi(\vec{a})} \vectorize\left(\Phi_\varepsilon - \Phi_\text{Haar}\right) \ket{\Psi(\vec{b})}\right|
\end{gather}
so that $\epsilon = \max(\epsilon_\text{even}, \epsilon_\text{odd})$. 

Then by the fact above, the maxima in  $\epsilon_\text{even}$ and $\epsilon_\text{odd}$ are saturated by $\vec{a} = \vec{b}$. Since the diagonal elements are all positive, we may now drop the absolute value to write
\begin{gather}
    \mathcal{M}\left(\Phi_\varepsilon, \Phi_\text{Haar}\right) = \max_{p \in \text{even, odd}} \left[\frac{1 + (-1)^p q^{-N}}{2} \max_{\mathcal{P}(\vec{a}) = p}\bra{\Psi(\vec{a})} \vectorize\left(\Phi_\varepsilon - \Phi_\text{Haar}\right) \ket{\Psi(\vec{a})}\right]
\end{gather}
This is the core formula on which we will rely for our computations. Theorem~\ref{thm:multerr_expdef} follows after re-replacing the Haar channel eigenvalue from Equation~\ref{eq:haar_definition} and inserting an extra copy of $\langle \Psi(\vec{a})|\Psi(\vec{a})\rangle = 1$.

\subsection{Experimental interpretation}
We can undo the vectorization map to recover one experimental interpretation of this formula. We have
\begin{align}
    \vectorize^{-1}\left(\ket{\widetilde{I}} \pm \ket{\widetilde{S}}\right) &\propto \left(T_I - \frac{1}{q} T_S\right) \pm \left(T_S - \frac{1}{q} T_I\right) \\
    &\propto 
     \left(1 \mp \frac{1}{q}\right)\left(T_I \pm T_S\right)
     \\
    \vectorize^{-1}\left(\ket{\Psi(\vec{a})}\right) &\propto \bigotimes_i \left(I + a_i T_S\right) \\
\end{align}
This is not a physical density matrix, since it is not normalized, but we can freely multiply by scalars without changing the ratio which appears in $\mathcal{M}$. 

The corresponding states, however, are more complicated than necessary. Rather than preparing a density matrix proportional to $I \pm T_S$, we may prepare any state whose projection into the commutant of the single-site unitary group is proportional to $I \pm T_S$. One may show that $\frac{I \pm S}{2}$ are a pair of commuting, orthogonal projectors into the symmetric and antisymmetric subspaces, respectively, under exchange of the two copies of the Hilbert space. In order to prepare $\ket{\Psi(\vec{a})}$, it thus suffices to prepare any state that is symmetric on sites where $a_i = +1$ and antisymmetric elsewhere. 

A simple choice of such states is $\ket{00}$ and $\ket{01} - \ket{10}$, respectively. If we prepare these states and make the corresponding projective measurements after the channel has acted, we obtain an experiment which saturates the multiplicative error. 

\subsection{Bound for tenuously-connected structures}
\label{app:disconnected}
We first prove a general result about disconnected circuits. Consider a nondeterministic architecture $\mathcal{E}$ satisfying the assumptions of Theorem \ref{thm:multerr_expdef}. Suppose there exists a cut $C$ with $m$ qudits on one side and $n-m$ on the other, such that with probability $p$ no gate crosses $C$ .
\begin{lemma}
    An ensemble $\Phi_\mathcal{E}$ which is disconnected with probability $p$ has multiplicative error at least
    \begin{gather}
    \mathcal{M}\left(\Phi_\mathcal{E}, \Phi_\text{Haar}\right) \geq  
    p\frac{(q^m + 1)(q^m + q^n)}{(q^m - 1)(q^n - q^m)}
    \end{gather}
\end{lemma}
\begin{proof}
By conditioning on connectedness, we may decompose $\Phi_{\mathcal{E}} = (1 - p)\Phi_C + p\Phi_{\slashed{C}}$. Let us compose $\Phi_{\slashed{C}}$ with a tensor product of Haar-random unitaries acting on all the qudits on each side of the cut to obtain $\Phi_{\text{Haar } m} \otimes \Phi_{\text{Haar }(n-m)}$. By the results of ref.~\cite{belkin_absence_2025}, this composition cannot increase the multiplicative error. Similarly, we may compose $\Phi_C$ with $\Phi_\text{Haar}$ to obtain $\Phi_\text{Haar}$. This shows
\begin{gather}
    \mathcal{M}(\Phi_{\mathcal{E}}, \Phi_{\text{Haar } n}) \geq \mathcal{M}((1-p)\Phi_{\text{Haar } n} + p \Phi_{\text{Haar } m} \otimes \Phi_{\text{Haar }(n-m)}, \Phi_{\text{Haar } n})
\end{gather}
The rest of the proof is then a straightfoward application of Theorem \ref{thm:multerr_expdef}. We compute
\begin{gather}
    \tr \left(\rho_{\vec{a}} \left[\Phi_{\text{Haar, } m} \otimes \Phi_{\text{Haar, }(n-m)}\right](\rho_{\vec{a}})\right) = \frac{2}{1 + (-1)^{a_1} q^{-m}} \frac{2}{1 + (-1)^{a_2} q^{-(n-m)}}
\end{gather}
and
\begin{gather}
    \tr \left(\rho_{\vec{a}} \Phi_{\text{Haar, } n}(\rho_{\vec{a}})\right) = \frac{2}{1 + (-1)^{a_1 + a_2} q^{-n}}
\end{gather}
The ratio is maximized by the choice $a_1 = a_2 = 1$, giving
\begin{align}
    \mathcal{M} &\geq p\left(\frac{1 + q^{-n}}{2} \frac{2}{1 -q^{-m}} \frac{2}{1 -q^{-(n-m)}} - 1\right) \\
    & \geq
    p\frac{(q^m + 1)(q^m + q^n)}{(q^m - 1)(q^n - q^m)}
\end{align}
\end{proof}

For the special case $m = \frac{n}{2}$ this simplifies to 
\begin{gather}
    \mathcal{M}\left(\Phi_\mathcal{E}, \Phi_\text{Haar}\right) \geq  
    p\left(\frac{q^{n/2} + 1}{q^{n/2} - 1}\right)^2 \geq p
\end{gather}
We are now ready to prove Theorem \ref{thm:disconnected}
\begin{proof}
    The bridge graph has $2 {n/2 \choose 2} + 1 = n^2/4 - n/2 + 1$ edges, so we can lower-bound the probability it is connected after $s$ gates as 
\begin{gather}
    P(\text{Bridge is connected} \leq \left(1 - \frac{1}{\frac{n^2}{4} - \frac{n}{2} + 1}\right)^s = \left(1 + \frac{4}{n(n-2)}\right)^{-s}
\end{gather}
The circuit size required to reach multiplicative error $\epsilon$ is thus lower-bounded by
\begin{align}
    s(\epsilon) 
    &\geq 
    \frac{\log \frac{1}{\epsilon}}{\log\left(1 + \frac{4}{n(n-2)}\right)}
    \\ & \geq \frac{n(n-2)}{4} \log \frac{1}{\epsilon}
\end{align}
\end{proof}

\subsection{Scaling the depth} \label{app:scaling_depth}
Consider composing together $s$ unitaries sampled independently from a distribution $\varepsilon$. The corresponding moment operator vectorizes to $\vectorize(\Phi_\varepsilon)^s$. If we again suppose $\vectorize(\Phi_\varepsilon)$ is positive-semidefinite, it has (unique) eigenvalues $\lambda_i$ and projections into eigenspaces $P_i$. The dominant eigenvalue is $\lambda_0 = 1$, and if our ensemble eventually approaches the Haar measure, then the corresponding eigenspace $P_0$ must be exactly the image of $\vectorize(\Phi_\text{Haar})$. We see
\begin{gather}
    \bra{\Psi(\vec{a})} \left(\vectorize\left(\Phi_\varepsilon\right)^s - \vectorize \left(\Phi_\text{Haar}\right) \right)\ket{\Psi(\vec{a})} = \sum_{i > 0} \lambda_i^s \left|\left|P_i \ket{\Psi(\vec{a})}\right|\right|^2
\end{gather}
i.e. the scaling depends only on how the norm of $\ket{\Psi(\vec{a})}$ decomposes into the eigenspaces of the vectorized moment operator. Attempting to maximize this distance over all $\vec{a}$ then gives the expression (\ref{eq:mult_error_over_depth}) for the multiplicative error.

\subsection{Known lower bounds via anticoncentration}
\label{app:dalzell_bounds}
Here we rephrase two bounds from ref.~\cite{dalzell_random_2022} in terms of our notation. 

An architecture with collision probability $Z$ cannot be an $\epsilon$-approximate $2$-design for any $\epsilon < \frac{Z}{Z_H} - 1$. Theorem 5 of ref.~\cite{dalzell_random_2022} shows that for the brickwork, circuit with $s$ gates,
\begin{gather}
    Z \geq \frac{Z_H}{2} \exp\left(A e^{\log n - 2a \frac{s}{n}}\right)
\end{gather}
with $A = \frac{1}{8 ce}$ and $c = 3e^{10}$. 

A brickwork with $s$ gates has depth $\frac{2s}{n}$. We can rearrange the lower bound into 
\begin{align}
    d &\geq \frac{\log n + \log A - \log\left(\log\left(2(1 + \epsilon\right)\right)}{\log \frac{q^2 + 1}{2q}}
    \\ &\approx \frac{\log n - 13.81}{\log \frac{q^2 + 1}{2q}}
\end{align}
where in the second line we've taken $\epsilon \rightarrow 0$. Note that this bound is only nontrivial above $N \sim 10^6$. 

From Theorem 4 of ref.~\cite{dalzell_random_2022}, on the other hand, an arbitrary architecture composed of Haar-random $2$-site gates cannot anticoncentrate to accuracy $\epsilon$ unless the gate count satisfies
\begin{gather}
    \frac{2s}{n} \geq \frac{\log n - \log \frac{(q+1) \log \left(1 + 2 \epsilon\right)}{\log(q + 1)} }{\log (q^2 + 1)}
\end{gather}
We may relax this to
\begin{gather}
    \frac{2s}{n} \geq \frac{\log n + \log \frac{1}{\epsilon} - \log \frac{2(q+1)}{\log(q + 1)} }{\log (q^2 + 1)}
\end{gather}
or, taking $q = 2$, 
\begin{gather}
    \frac{2s}{n} \geq \log_5 \frac{n}{\epsilon} - 0.801
\end{gather}
This gives a lower bound on the $\epsilon$-approximate $2$-design depth for an arbitrary architecture.

\section{Counting connections} \label{app:connections_of_architectures}
\subsection{Defining connection count}
Section \ref{sec:graphs} defines connectedness in a somewhat subtle way. This definition is motivated by ref.~\cite{belkin_approximate_2024} and has a few convenient properties. A succinct statement is as follows:
\begin{definition}
    The \textbf{connection count} of a fixed random quantum circuit architecture $A$ is the largest number of \textit{connected blocks} into which any architecture \textit{equivalent} to $A$ can be divided.
\end{definition}
Here by a \textbf{connected block} of an architecture we mean a sequence of consecutive gates which form a connected graph over all of the qubits.  By \textit{fixed} we mean a particular arrangement of gates - the local unitaries themselves are permitted to be random, but their locations are not. So this definition itself doesn't apply to most of the architectures studied here. For a nondeterministic architecture, such as one sampled from a graph, we instead concern ourselves with the \textbf{mean connection count}, which is just the average connection count over all of its realizations. 

Two architectures are \textbf{equivalent} if they induce the same measure on the unitary group. Let us represent our architecture by an ordered list of pairs of qubits, each corresponding to a gate location. In practice we care about the following two rules:
\begin{enumerate}
    \item If two consecutive gates act in disjoint locations, then we can swap their ordering. 
    \item We may split any gate into two consecutive copies of itself. 
\end{enumerate}

For example, given four qubits labeled $a,b,c,d$, the following two architectures are equivalent:
$$ab, ad, bc$$
$$ab, ad, bc, ad$$

\subsection{Naive and greedy algorithms}
We do not know of a guaranteed way to compute the connection count, as defined above, since there may be many possible ways to rearrange and slice up an architecture. However, we can compute lower bounds. The first approach we consider is a naive algorithm which doesn't inspect equivalent architectures at all. We simply add gates to the current block until it becomes connected, then slice that block off and proceed. 

The results of this algorithm are relatively easy to analyze. For example, the naive mean connection count of the complete graph with $s$ gates on $n$ qubits is $\frac{2s}{n \log n + O(n)}$, by percolation. For the star and linear graphs, it's a coupon collector problem, so we get mean connection count $\frac{s}{n \log n + O(n)}$. This is illustrated for the linear graph in Figure \ref{fig:greedy_vs_naive}. 

On the other hand, we can make some attempt to use the three reduction rules above to reduce this number. We use a greedy algorithm, which proceeds as follows: 
\begin{itemize}
    \item Add gates to the current block until it becomes connected.
    \item For each gate in the last layer of the current block, i.e. each gate which commutes with every gate after it, check if it can be removed without disconnecting the block. If it can, remove it from the current block and add it to the next block.
    \item Duplicate the last layer of the current block and add it as the first layer of the next block.
\end{itemize} 
These reductions give in general much higher connection counts, since many gates can be used twice. For example, Figure \ref{fig:greedy_vs_naive} shows that the greedy algorithm finds $\sim \frac{s}{2.8 n}$ connections for a typical set of $s$ gates sampled from the linear graph on $n$ qubits, vs $\sim \frac{s}{n \log n}$ found by the naive algorithm.

The case of the linear graph illustrates why this more complicated definition of connectedness is interesting. The total circuit size required to form an approximate $2$-design appears to be $\sim 12.3 n \log n$, which corresponds to only $\sim 12.3 $ connections under a naive count. Under a greedy count, however, we find that the first connected block requires $\sim n \log n + \gamma n$ gates, but at very large $n$ we see subsequent blocks have only $\sim 2.6 n$ gates. This suggests that the total number of connections needed may be asymptotically closer to $\sim 4.3 \log n$. From Figure \ref{fig:basic_connections}, however, we can see that this scaling wouldn't kick in until around $100$ qubits. 

\begin{figure}[h]
    \centering
    \begin{tikzpicture}
        \begin{scope}
            \node[anchor=north west,inner sep=0] (image_a) at (0,0)
            {\includegraphics[width=0.6\columnwidth]{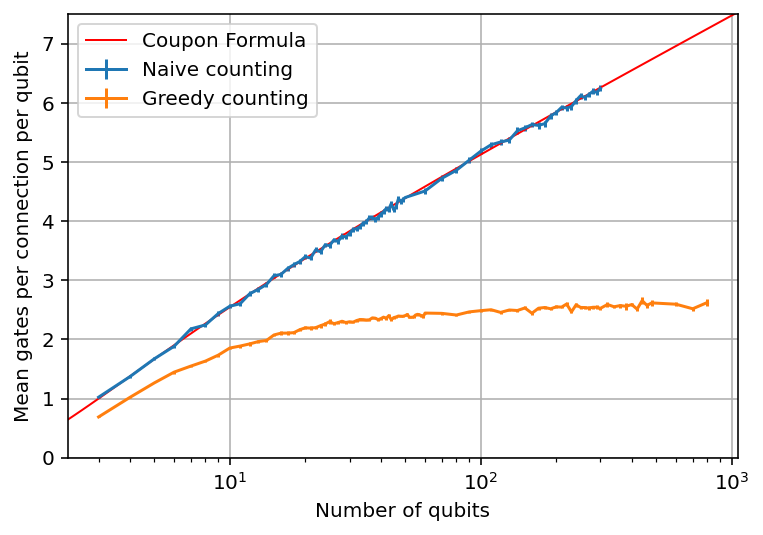}};
        \end{scope}
    \end{tikzpicture}
    \vspace*{-0.4cm}
    \caption{Mean gates per connection per qubit for the linear graph, as counted by the naive and greedy methods. The naive method matches the theoretical prediction, implying $\Theta(n \log n)$ gates per connection. The greedy method, on the other hand, suggests only $\Theta(n)$ gates per connection.}
    \label{fig:greedy_vs_naive}
\end{figure}

The average connection count found by the greedy algorithm does not grow linearly with depth. Early blocks gain relatively few ``free'' gates from their predecessors, and so the connection count grows relatively slowly initially. At large depths it asymptotes to a constant growth rate, as illustrated in Figure \ref{fig:connections_vs_gates}.

\begin{figure}[h]
    \centering
    \begin{tikzpicture}
        \begin{scope}
            \node[anchor=north west,inner sep=0] (image_a) at (0,0)
            {\includegraphics[width=0.6\columnwidth]{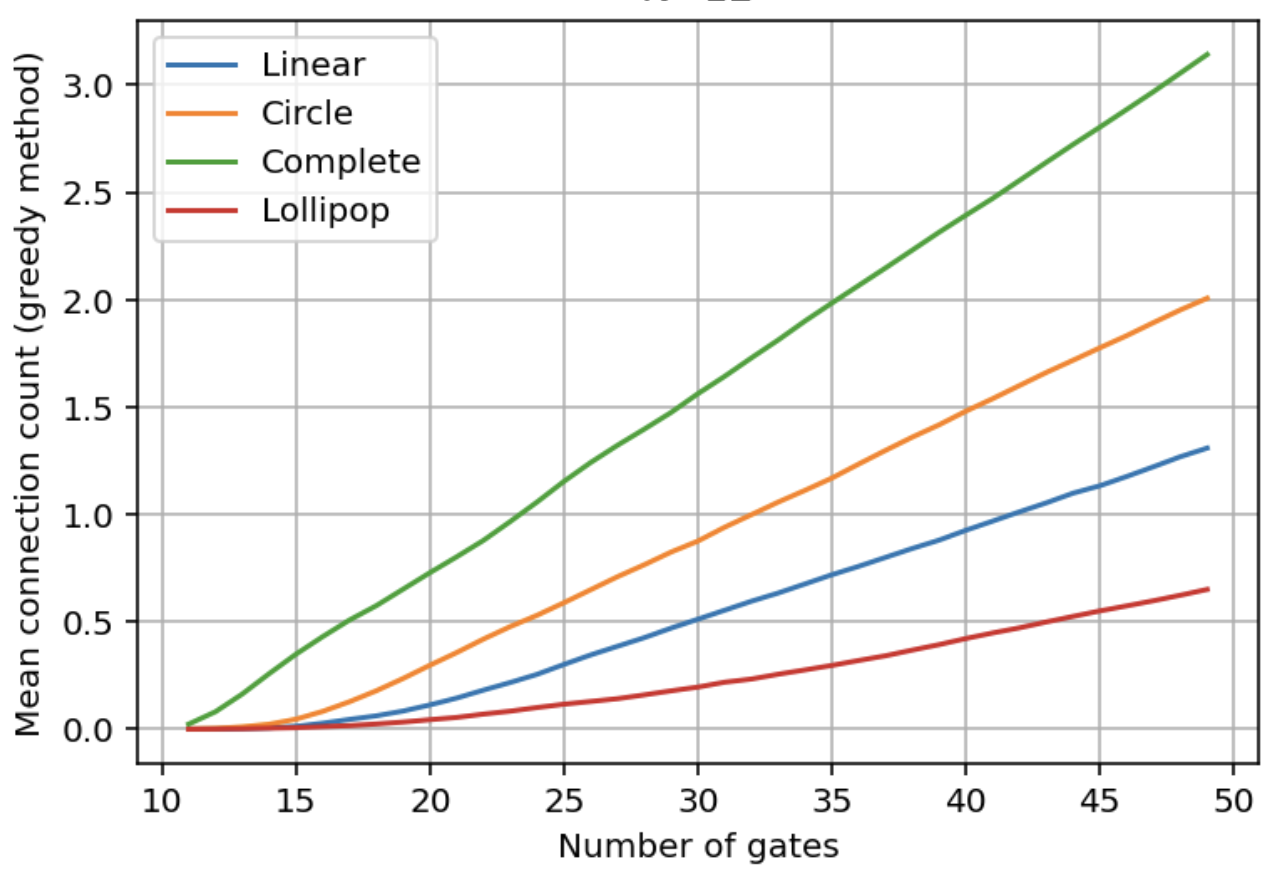}};
        \end{scope}
    \end{tikzpicture}
    \vspace*{-0.4cm}
    \caption{Estimated mean connection count vs. circuit size for each of four graphs on 12 qubits.}
    \label{fig:connections_vs_gates}
\end{figure}

\section{Improving computational complexity}
\label{app:algo_tricks}
Here we list several tricks we used to make the numerical calculations above tractable. 
\paragraph{Choice of basis} Our goal is to evaluate 
\begin{gather}
    \max_{\vec{a}} \bra{\Psi(\vec{a})} \vectorize\left(\Phi_\varepsilon\right)^d \ket{\Psi(\vec{a})}
\end{gather}
numerically. 
For sufficiently small systems, one may work out the transfer matrix $\vectorize \Phi_{\varepsilon}$ explicitly. In the physical basis, this matrix has $q^{2N}$ entries, so it becomes impractical quite quickly. It's more useful to use the permutation basis on one side and the cobasis on the other, i.e. resolve the orthogonal projector into the single-site commutant as $P_\text{comm} = \sum_{\sigma} \ket{\sigma}\bra{\widetilde{\sigma}}$. We then may define
\begin{align}
    H_{\vec{\sigma}, \vec{\tau}} &= \bra{\widetilde{\vec{\sigma}}} \vectorize \Phi_{\varepsilon} \ket{\vec{\tau}}
    \\
    \mathbf{L}_{\vec\sigma}(\vec{a}) &= \braket{\Psi(\vec{a})|\vec{\sigma}}
    \\
    \mathbf{R}_{\vec\sigma}(\vec{a}) &= \braket{\widetilde{\vec{\sigma}}|\Psi(\vec{a})}
\end{align}
so that 
\begin{gather}
    \bra{\Psi(\vec{a})} \vectorize\left(\Phi_\varepsilon\right)^d \ket{\Psi(\vec{a})} = \mathbf{L}^T(\vec{a}) H^d \mathbf{R}(\vec{a})
\end{gather}
Note that even though it comes from a Hermitian operator, $H$ is not a symmetric matrix when expressed in this non-orthogonal basis. 

\paragraph{Tracking fewer irreps} Our goal is to find an $\epsilon$-approximate $2$-design depth. On option is to work out the formula above for all $\vec{a}$, increasing $d$ until the maximum is reached. However, this requires $2^N$ choices of $\vec{a}$, and only a few will contribute to the maximum anywhere. A better strategy is to observe that for each $\vec{a}$, the quadratic form above is a monotonically decreasing function of $d$. For most choices of $\vec{a}$ it will decrease very rapidly, so we need to consider only very shallow circuits. For any given $\epsilon$, typically there are only a few choices of boundary state which need to be carried to large depth. 

\paragraph{State representation} A second observation is that it is typically not necessary to evaluate the $4^N$ entries of $H$ explicitly. Usually one can instead compute the action of $H$ on a desired vector directly, storing either $2^N$ vector elements for an unstructured vector or some structured tensor-network representation of the state.

\paragraph{Interpolation} For the brickwork, the $\epsilon$-approximate $2$-design depth is a discrete quantity. This makes finite-$n$ behavior rather messy. To give more useful insight into the actual scrambledness, we use interpolation on $\log \epsilon$ to obtain a continuous value. We do the same for the fast architectures. 

\paragraph{Tricks for graphs}
Consider the case of Haar-random 2-site unitaries acting on sites edges sampled uniformly from the edges of some graph. In this case the transfer matrix may be expressed as 
\begin{gather}
    \frac{1}{|E|} \sum_{(i,j) \in E} G_{ij} \otimes I_{d-2}
\end{gather}
where $G_{ij}$ is the local moment operator corresponding to a Haar-random gate acting on sites $i$ and $j$. 
Since each gate is small, it's easy to compute $G_{ij}\ket{\psi}$ for each choice of edge and then sum. This has runtime $O(2^N |E|)$. 

In addition, graphs often have symmetries. Any two choices of $\vec{a}$ which are related to each other by an automorphism of the graph will give the same contribution to the multiplicative error, so we need choose only one representative for each automorphism class.

\paragraph{Tricks for brickwork}
For a fixed arrangement of Haar-random unitaries, there's an additional simplification which halves the effective system size. Rather than being invariant under single-site unitaries, our measure is now invariant under two-site unitaries on those adjacent sites paired by the circuit. It follows that the commutant into which we are projected is of dimension $(t!)^{N/2}$ instead of $(t!)^{N}$. This corresponds to a singular value decomposition of each two-site local moment operator.

\paragraph{Tricks for fast architectures}
The fast architectures studied each have permutation symmetry, so that only $O(n)$ distinct choices of $\vec{a}$ must be considered. For these architectures the number possible configurations of each layer grows factorially with $n$, and so an exact evaluation of the moment operator is not tractable. We instead sample over circuit realizations and average together the results. This gives a consistent estimator for the true multiplicative error. The PBFE also has sites which are paired in a predictable way. This allows us to represent the state with a tensor with only $2^{n/2}$ elements, which makes computations tractable for twice as many qubits. 

\end{appendices}
\end{document}